\RequirePackage{fix-cm}
\documentclass[twocolumn]{svjour3}       
\smartqed  
\usepackage{graphicx, hyperref, breakurl}
\usepackage{epstopdf}
\usepackage{natbib, url ,amsmath, amssymb, latexsym, appendix}
\usepackage{float} 
\floatstyle{plain}
\newfloat{Algorithm}{thp}{lop}
\floatname{Algorithm}{Algorithm}
\newenvironment{wideenumerate}{\enumerate\addtolength{\itemsep}{5pt}}{\endenumerate}
\newcommand*{\QEDB}{\hfill\ensuremath{\square}}
\begin{document}
\sloppy

\title{Variational inference for sparse spectrum Gaussian process regression}


\author{Linda S. L. Tan  \and  Victor M. H. Ong \and David J. Nott \and Ajay Jasra }


\institute{Linda S. L. Tan \at
              Department of Statistics and Applied Probability \\
              National University of Singapore \\
                            Tel.: +65-6516-4416\\ 
                            Fax: +65-6872-3919\\  
              \email{statsll@nus.edu.sg}           
\and
              Victor M. H. Ong \at
               Department of Statistics and Applied Probability \\
              National University of Singapore \\
              \email{victor84@u.nus.edu.sg}           
\and
           David J. Nott \at
           Department of Statistics and Applied Probability \\
           National University of Singapore \\
           \email{standj@nus.edu.sg} 
\and
          Ajay Jasra \at
           Department of Statistics and Applied Probability \\
           National University of Singapore \\
           \email{staja@nus.edu.sg}
}

\date{Received: date / Accepted: date}

\maketitle

\begin{abstract} 
We develop a fast variational approximation scheme for Gaussian process (GP) regression, where the spectrum of the covariance function is subjected to a sparse approximation. Our approach enables uncertainty in covariance function hyperparameters to be treated without using Monte Carlo methods and is robust to overfitting. Our article makes three contributions. First, we present a variational Bayes algorithm for fitting sparse spectrum GP regression models that uses nonconjugate variational message passing to derive fast and efficient updates. Second, we propose a novel adaptive neighbourhood technique for obtaining predictive inference that is effective in dealing with nonstationarity. Regression is performed locally at each point to be predicted and the neighbourhood is determined using a measure defined based on lengthscales estimated from an initial fit. Weighting dimensions according to lengthscales, this downweights variables of little relevance, leading to automatic variable selection and improved prediction. Third, we introduce a technique for accelerating convergence in nonconjugate variational message passing by adapting step sizes in the direction of the natural gradient of the lower bound. Our adaptive strategy can be easily implemented and empirical results indicate significant speedups.

\end{abstract}

\section{Introduction}
Gaussian process (GP) models provide a flexible, probabilistic approach to regression and are widely used. However, application of GP models to large data sets is challenging as the memory and computational requirements scale as $O(n^2)$ and $O(n^3)$ respectively, where $n$ is the number of training data points. Various sparse GP approximations have been proposed to overcome this limitation. A unifying framework of existing sparse methods is given in \cite{Quinonero2005}. We consider the stationary sparse spectrum GP regression model introduced by \cite{Lazaro2010}, where the spectrum of the covariance function is sparsified instead of the usual spatial domain. The SSGP algorithm developed by \cite{Lazaro2010} for fitting this model uses conjugate gradients to optimize the marginal likelihoood with respect to the hyperparameters and spectral points. Comparisons with other state-of-the-art sparse GP approximations such as the fully independent training conditional model \citep[first introduced as sparse pseudo-input GP in][]{Snelson2006} and the sparse multiscale GP \citep{Walder2008}, showed that SSGP yielded significant improvements. However, optimization with respect to spectral frequencies increases the tendency to underestimate predictive uncertainty and poses a risk of overfitting in the SSGP algorithm.

In this paper, we develop a fast variational approximation scheme for the sparse spectrum GP regression model, which enables uncertainty in covariance function hyperparameters to be treated. In addition, we propose an adaptive local neighbourhood approach for dealing with nonstationary data. Although accounting for hyperparameter uncertainty may be of little importance when fitting globally to a large data set, local fitting within neighbourhoods results in fitting to small data sets even if the full data set is large, and here it is important to account for hyperparameter uncertainty to avoid overfitting.  Our examples show that our methodology is particularly beneficial when combined with the local fitting approach for this reason. Our approach also allows hierarchical models involving covariance function parameters to be constructed. This idea is implemented in the context of functional longitudinal models by Mensah {\it et al.} (2014) so that smoothness properties of trajectories can be related to individual specific covariates.

GPs have diverse applications and various methods have been developed to overcome their computational limitations for handling large data sets. A good summary of approximations used in modelling large spatial data sets is given in \cite{Ren2011}. Computational costs can also be reduced through local GP regression as a much smaller number of training data is utilized in each partition. This approach has been considered in machine learning \cite[e.g.][]{Snelson2007, Nguyen2009, Park2010} and in spatial statistics \citep[e.g.][]{Vecchia1988, Haas1995, Stein2004, Kim2005}. \cite{Urtasun2008} propose fitting GP models in local neighbourhoods which are defined online for each test point. However, covariance hyperparameters are estimated only for a subset of all possible local neighbourhoods. Different local experts are then combined using a mixture model capable of handling multimodality. Our idea of using adaptive nearest neighbours in GP regression is inspired by techniques in classification designed to mitigate the curse of dimensionality \citep{Hastie1996}. For each test point, we fit two models. In the first instance, the neighbourhood is determined using the Euclidean metric. Lengthscales estimated from the first fitting are then used to redefine the distance measure determining the neighbourhood for fitting the second model. Experiments suggest that this approach improves prediction significantly in data with nonstationarities, as hyperparameters are allowed to vary across neighbourhoods adapted to each query point. Weighting dimensions according to lengthscales downweights variables of little relevance and also leads to automatic variable selection. Our approach differs from methods where local neighbourhoods are built sequentially to optimize the choice of the neighbourhood. Examples include \cite{Vecchia1988} and \cite{Stein2004}, where the Gaussian likelihood is approximated by the use of an ordering and conditioning on a subset of past observations. In \cite{Gramacy2014}, an empirical Bayes mean-square prediction error criterion is optimized. While greedy searches usually rely on fast updating formulae available only in the Gaussian case, our approach works in non-Gaussian settings as well. \cite{Stein2004} suggest making neighbourhoods non-local to improve learning of covariance parameters, but local neighbourhoods may work better when the motivation is to handle nonstationarity. \cite{Lindgren2011} make a connection between discrete spatial Markov random fields and continuous Gaussian random fields with covariance functions in the Mat\'{e}rn class.

For fitting the sparse spectrum GP regression model, we derive a variational Bayes \citep[VB,][]{Attias1999} algorithm that uses nonconjugate variational message passing \citep{Knowles2011} to derive fast and efficient updates. VB methods approximate the intractable posterior in Bayesian inference by a factorized distribution. This product density assumption is often unrealistic and can lead to underestimation of posterior variance \citep{Wang2005}. However, optimization of a factorized variational posterior can be decomposed into local computations that only involve neighbouring nodes in the factor graph and this often gives rise to fast computational algorithms. VB has also been shown to be able to give reasonably good estimates of the marginal posterior distributions and excellent predictive inferences \citep[e.g.][]{Blei2006, Braun2010}. Variational message passing \citep{Winn2005} is a general-purpose algorithm that allows VB to be applied to conjugate-exponential models \citep{Attias2000}. Nonconjugate variational message passing extends variational message passing to nonconjugate models by assuming that the factors in VB are members of the exponential family. We use nonconjugate variational message passing to derive efficient updates for the variational posteriors of the lengthscales, which are assumed to be Gaussian. \cite{Ren2011} use VB for spatial modelling via GP, where they also treat uncertainty in the covariance function hyperparameters. However, they propose using importance sampling within each VB iteration to handle the intractable expectations associated with the covariance function hyperparameters. Variational inference has also been considered in machine learning for sparse GPs that select the inducing inputs and hyperparameters by maximizing a lower bound to the exact marginal likelihood \citep{Titsias2009}, and heteroscedastic GP regression models where the noise is input dependent \citep{Lazaro2011}.

VB is known to suffer from slow convergence when there is strong dependence between variables in the factors. To speed up convergence, \cite{Qi2006} propose parameter expanded VB to reduce coupling in updates, while \cite{Tan2013} considered partially noncentered parametrizations. Here, we introduce an adaptive strategy to accelerate convergence in nonconjugate variational message passing, which is inspired by adaptive overrelaxed bound optimization methods \citep{Salak2003}. Previously, \cite{Tan2014} showed that nonconjugate variational message passing is a natural gradient ascent algorithm with step size one and step sizes smaller than one correspond to damping. Here, we propose using step sizes larger than one which can help to accelerate convergence in fixed point iterations algorithms \citep[see][]{Huang2005}. Instead of searching for the optimal step size, we use an adaptive strategy which ensures that the lower bound increases after each cycle of updates. Empirical results indicate significant speedups. \cite{Honkela2003} considered combining parameter-wise updates to form a diagonal direction for a line search. A general iterative algorithm for computing VB estimators (defined as means of variational posteriors) has also been proposed by \cite{Wang2006} and its convergence properties investigated for normal mixture models.

Section \ref{SSGPmodel} describes the sparse spectrum GP regression model and Section \ref{variational inference} develops the nonconjugate variational message passing algorithm for fitting it. Section \ref{adaptive strategy} presents an adaptive strategy for accelerating convergence in nonconjugate variational message passing. Section \ref{pred distn} discusses how the predictive distribution can be estimated and the measures used for performance evaluation. Section \ref{neigh} describes the adaptive neighbourhood approach for local regression. Section \ref{eg} considers examples including real and simulated data and Section \ref{conclusion} concludes.

\section{Sparse spectrum Gaussian process regression}\label{SSGPmodel}
Given a data set $\{(x_i,y_i)|i=1,\dots,n\}$, we assume each output $y_i\in \Re$ is generated by an unknown latent function $f$ evaluated at the input, $x_i \in \Re^d$, and independently corrupted by additive Gaussian noise such that
\begin{equation*}
y_i=f(x_i)+\epsilon_i, \quad  \epsilon_i \sim N(0,\gamma^2).
\end{equation*}
A GP prior is assumed over $f(x)$ for $x \in \Re^d$. For any set of inputs $\{x_i|i=1,\dots,n\}$, $[f(x_1),\dots,f(x_n)]^T$ has a joint Gaussian distribution, $N(0,K)$, where $K$ is a covariance matrix. We assume that the mean of the process is zero. It is straightforward to allow for a nonzero mean, but a zero mean is sufficient for the examples in this paper. The entries of $K$ are given by $K_{ij}=E\{f(x_i)f(x_j)\}=k(x_i,x_j)=k(h)$, where $h=(x_i-x_j) \in \Re^d$ and $k$ is some stationary covariance function. For example, we consider the stationary squared exponential covariance function,
\begin{equation}\label{sqexpcov}
k(h)=\sigma^2\exp(-\tfrac{1}{2}h^T\Lambda h), 
\end{equation}
where $\sigma^2>0$, $\Lambda = \text{diag} ([\lambda_1^2, \dots, \lambda_d^2]^T)$ and $\lambda_l\geq 0$ for $l=1,\dots,d$. 

\citet{Lazaro2010} introduced a novel perspective on GP approximation by sparsifying the spectrum of the covariance function. They considered the linear regression model, 
\begin{equation}\label{spec}
f(x) \approx \sum_{r=1}^m \left\{a_r\cos(2\pi s_r^T x)+b_r \sin(2\pi s_r^T x)\right\}, 
\end{equation}
where $a_r$, $b_r$ are independent and identically distributed as $N(0,\frac{\sigma^2}{m})$ and $s_r$ is a $d$-dimensional vector of spectral frequencies. The power spectral density of a stationary covariance function $k$ is 
\begin{equation}\label{spectral density}
S_k(s)=\int_{\Re^d} \exp(-2\pi i s^T h)k(h)\;dh,
\end{equation}
and $S_k(s)$ is proportional to a probability density $p_k(s)$ such that $S_k(s)=k(0)p_k(s)$. When $\{s_1,\dots,s_m\}$ are drawn randomly from $p_k(s)$, \citet{Lazaro2010} showed that \eqref{spec} can be viewed as a sparse GP that approximates the full stationary GP by replacing the spectrum with a discrete set of spectral points. 

From  \eqref{spectral density}, the probability density $p_k(s)$ associated with the squared exponential covariance function in \eqref{sqexpcov} is $N(0,\frac{1}{4\pi^2}\Lambda^{-1})$. If $\{s_1,\dots,s_m\}$ is generated randomly from $N(0,I_d)$, then $\{\tfrac{1}{2\pi}\Lambda^{\frac{1}{2}}s_1,\dots,\tfrac{1}{2\pi}\Lambda^{\frac{1}{2}}s_m\}$ is a random sample from $p_k(s)$. From \eqref{spec}, a sparse GP approximation to $f(x)$ is 
\begin{align}
f(x)&\approx\sum_{r=1}^m \left\{a_r\cos(s_r^T \Lambda^{\frac{1}{2}} x)+b_r \sin(s_r^T \Lambda^{\frac{1}{2}} x)\right\} \nonumber \\
&=\sum_{r=1}^m \left[a_r\cos\{(s_r\odot x)^T\lambda \}+b_r \sin\{(s_r\odot x)^T\lambda \}\right], \label{spec2}
\end{align}
where $\Lambda^{\frac{1}{2}}=\text{diag}(\lambda)$, $\lambda = [\lambda_1, \dots, \lambda_d]^T$ is a vector of lengthscales and $\odot$ denotes element by element multiplication of two vectors. Within the sparse GP approximation, we can allow the components of $\lambda$ to be negative. Let $s=[s^1, \dots, s^d]^T$ and $x=[x^1, \dots, x^d]^T$. Note that in (\ref{sqexpcov}) $\lambda_j$ appears as its square in $\Lambda$ so that $k(h)$ remains positive semidefinite. Ignoring the non-negativity constraint allows us to use a Gaussian variational posterior for $\lambda$. The associated expectations in the variational lower bound can then be derived in closed form (see Section \ref{variational inference}). This is a highly novel aspect of our algorithm allowing a fast method that still handles covariance function hyperparameter uncertainty. This is especially important when fitting locally as described in Section 6 where training datasets may be small. The squared exponential covariance function also implements automatic relevance determination since the magnitude of $\lambda_j$ is a measure of how relevant the $j$th variable is. When $\lambda_j$ goes to zero, the covariance function becomes almost independent of the $j$th variable, essentially removing it from inference. See \cite{Rasmussen2006} for more discussion.

Using the stationary sparse GP approximation in (\ref{spec2}), we consider variational inference for
\begin{multline*}
y_i=\sum_{r=1}^m \left[a_r\cos\{(s_r\odot x_i)^T\lambda\}+b_r \sin\{(s_r\odot x_i)^T\lambda \}\right] \\ +\epsilon_i, \;\;\text{where}\;\; \epsilon_i \sim N(0,\gamma^2).
\end{multline*}
Let $\alpha=[a_1,\dots,a_m,b_1,\dots,b_m]^T$, $y=[y_1,\dots,y_n]^T$, $\epsilon=[\epsilon_1,\dots,\epsilon_n ]^T$ and $Z=[
Z_1,\dots,Z_n]^T$, where 
\begin{multline*}
Z_i=[ \cos\{(s_1\odot x_i)^T\lambda \},\dots,\cos\{(s_m\odot x_i)^T\lambda \}, \\ 
\sin\{(s_1\odot x_i)^T\lambda \},\dots,\sin\{(s_m\odot x_i)^T\lambda \}]^T.
\end{multline*}
Then this model can be written as 
\begin{equation}\label{model2}
y=Z\alpha + \epsilon, \quad \epsilon \sim N(0,\gamma^2 I_n), 
\end{equation}
where $\alpha\sim N(0,\frac{\sigma^2}{m}I_{2m})$.  For Bayesian inference, we assume the priors: $\lambda \sim N(\mu_\lambda^0, \Sigma_\lambda^0)$, $\sigma \sim \text{half-Cauchy}(A_\sigma)$ and $\gamma \sim \text{half-Cauchy}(A_\gamma)$, where the hyperparameters $\mu_\lambda^0$, $\Sigma_\lambda^0$, $A_\sigma$ and $A_\gamma$ are assumed to be known. The density function of a random variable $x$ distributed as $\text{half-Cauchy}(A)$ is $\tfrac{2A}{\pi(A^2+x^2)}$, where $x>0$ and $A>0$. While inverse-Gamma priors are more commonly used for variance parameters in hierarchical models due to the conditional conjugacy relationship with Gaussian families, \cite{Gelman2006} recommends use of the half-Cauchy family as priors because resulting inferences can be sensitive to inverse-Gamma hyperparameters when variance estimates are close to zero. We made the same observation in our experiments with inverse-Gamma priors for $\sigma^2$ and $\gamma^2$. In particular, predictive inferences are sensitive to inverse-Gamma priors in local regressions (see Section \ref{neigh}), where only a small neighbourhood is used for fitting at each test point.

\section{Variational inference}\label{variational inference}

We consider variational inference for the sparse spectrum GP regression model in (\ref{model2}). Let $\theta=\{\alpha,\lambda,\sigma,\gamma\}$ be the set of unknown parameters and $p(\theta|y)$ be the true posterior of $\theta$. In variational approximation, $p(\theta|y)$ is approximated by a $q(\theta)$ for which inference is more tractable, and the Kullback-Leibler divergence between $q(\theta)$ and $p(\theta|y)$ is minimized. This is equivalent to maximizing a lower bound $\mathcal{L}$ on the log marginal likelihood $\log p(y)$, where $p(y) = \int p(y,\theta)\,d\theta$,
\begin{equation}\label{LB}
\mathcal{L} = E_q\{\log p(y,\theta)\}-E_q\{\log q(\theta)\},
\end{equation}
and $E_q$ denotes expectation with respect to $q(\theta)$. 

Next, we review some important results in VB and nonconjugate variational message passing, which will be used to construct the variational algorithm. In VB, $q(\theta)$ is assumed to factorize into $\prod_{i=1}^M q_i(\theta_i)$ for some partition $\{\theta_1,\dots,\theta_M\}$ of $\theta$. The optimal densities may be obtained from
\begin{equation}\label{VB}
q_i(\theta_i) \propto \exp \{E_{-\theta_i}\log p(y,\theta)\}, \;\;  i=1,\dots,M,
\end{equation}
where $E_{-\theta_i}$ denotes expectation with respect to $\prod_{j\neq i}q_j({\theta_j})$ \citep[see, e.g.][]{Ormerod2010}. For conjugate-exponential models, the optimal densities have the same form as the priors and it suffices to update the parameters of $q_i$, such as in variational message passing \citep{Winn2005}. However, for nonconjugate models, the optimal densities will not belong to recognizable density families. Apart from the product assumption, nonconjugate variational message passing \citep{Knowles2011} further assumes each $q_i(\theta_i)$ is a member of some exponential family, that is, 
\begin{equation*}
q_i(\theta_i)=\exp\{\eta_i^T t_i(\theta_i)-h_i(\eta_i)\},
\end{equation*}
where $\eta_i$ is the vector of natural parameters and $t_i(\cdot)$ are the sufficient statistics. Hence, we only have to find each $\eta_i$ that maximizes the lower bound $\mathcal{L}$. Nonconjugate variational message passing can be interpreted as a fixed point iterations algorithm where updates are obtained from the condition that the gradient or natural gradient \citep[see][]{Amari1998, Hoffman2013} of $\mathcal{L}$ with respect to each $\eta_i$ is zero when $\mathcal{L}$ is maximized. Suppose $p(y,\theta)=\prod_a f_a(y,\theta)$, $S_a = E_q\{\log f_a(y,\theta)\}$ and let $\mathcal{V}_i(\eta_i)=\frac{\partial^2 h_i(\eta_i)}{\partial \eta_i \partial \eta_i^T}$ denote the variance-covariance matrix of $t_i(\theta_i)$. Provided $\mathcal{V}_i(\eta_i)$ is invertible, \cite{Tan2014} showed that the natural gradient of $\mathcal{L}$ with respect to $\eta_i$ is 
\begin{equation}\label{natgrad}
\tilde{\nabla}_{\eta_i}\mathcal{L} = \mathcal{V}_i(\eta_i)^{-1}\sum_{a \in N(\theta_i)} \frac{\partial S_a}{\partial \eta_i}- \eta_i.
\end{equation}
Therefore, the update for each $\eta_i$ is 
\begin{equation}\label{NCVMP_update}
\eta_i \leftarrow  \mathcal{V}_i(\eta_i)^{-1}\sum_{a \in N(\theta_i)} \frac{\partial S_a}{\partial \eta_i},
\end{equation}
where the summation is over all factors in $N(\theta_i)$, the neighbourhood of $\theta_i$ in the factor graph of $p(y,\theta)$. Updates in nonconjugate variational message passing reduce to those in variational message passing when the factors $f_a$ are conjugate \citep[see][]{Knowles2011, Tan2013}. However, unlike variational message passing, the lower bound $\mathcal{L}$ is not guaranteed to increase at each step and convergence problems may be encountered sometimes. \cite{Knowles2011} suggest using damping to fix convergence problems. 

When $q_i(\theta_i)=N(\mu_{\theta_i}^q,\Sigma_{\theta_i}^q)$, \cite{Wand2013}  showed that the update in \eqref{NCVMP_update} can be simplified to
\begin{equation}
\begin{aligned}\label{Gauss}
\Sigma_{\theta_i}^q &\leftarrow -\frac{1}{2}\bigg[\text{vec}^{-1}\bigg(\sum_{a \in N(\theta_i)} \negthickspace \frac{\partial S_a}{\partial  \text{vec}(\Sigma_{\theta_i}^q)}\bigg)\bigg]^{-1}, \\
\mu_{\theta_i}^q &\leftarrow \mu_{\theta_i}^q + \Sigma_{\theta_i}^q\sum_{a \in N(\theta_i)}\negthickspace \frac{\partial S_a}{\partial  \mu_{\theta_i}^q}. 
\end{aligned}
\end{equation}
Here $\text{vec}(A)$ denotes the vector obtained by stacking the columns of a matrix $A$ under each other, from left to right in order.

\subsection{Algorithm 1} \label{VASSGP}

We consider a variational approximation of the form 
\begin{equation}\label{pdt}
q(\theta)=q(\alpha)q(\lambda)q(\sigma,\gamma).
\end{equation}
From \eqref{VB}, the optimal densities $q(\alpha)$ and $q(\sigma,\gamma)$ are $q(\alpha) = N(\mu_{\alpha}^q, \Sigma_{\alpha}^q)$ and $q(\sigma,\gamma)=q(\sigma)q(\gamma)$, where
\begin{equation*}
\begin{aligned}
q(\sigma)&=\frac{\exp(-{C_\sigma^q}/{\sigma^2})}{\mathcal{H}(2m-2,C_\sigma^q,A_\sigma^2) \sigma^{2m}(A_\sigma^2+\sigma^2)},  \\
q(\gamma)&=\frac{\exp(-{C_\gamma^q}/{\gamma^2})}{\mathcal{H}(n-2,C_\gamma^q,A_\gamma^2) \gamma^{n}(A_\gamma^2+\gamma^2)},
\end{aligned}
\end{equation*}
and $\mathcal{H}(p,q,r)=\int_0^\infty  x^p\exp\{-qx^2-\log(r+x^{-2})\}\;dx$, $p\geq 0$, $r>0$. The variational parameter updates of $\mu_{\alpha}^q$, $\Sigma_{\alpha}^q$, $C_\sigma^q$ and $C_\gamma^q$ can also be derived from \eqref{VB}. As $\mathcal{H}(p,q,r)$ can be arbitrarily large or small, \citet{Wand2011} suggest evaluating $\log \mathcal{H}(p,q,r)$ efficiently using quadrature. A discussion can be found in Appendix B of \cite{Wand2011} and we follow their methods. For $q(\lambda)$, $p(y|\alpha,\lambda,\gamma)$ is not a conjugate factor and we use nonconjugate variational message passing. Assuming $q(\lambda)=N(\mu_\lambda^q,\Sigma_\lambda^q)$, updates for $\mu_\lambda^q$ and $\Sigma_\lambda^q$ can be derived using \eqref{Gauss} and matrix differential calculus \citep[see][]{Magnus1988}. The expectations with respect to $q$ in \eqref{Gauss} are given in Appendices A and B. Let $\vartheta=\{ \mu_{\alpha}^q, \Sigma_{\alpha}^q, \mu_\lambda^q, \Sigma_\lambda^q, C_\sigma^q, C_\gamma^q\}$ denote the set of variational parameters. An iterative scheme for finding $\vartheta$ is given in Algorithm \ref{Alg1}. 

\begin{Algorithm*}
\centering
\parbox{0.8\textwidth}{
\hrule
\vspace{1mm}
Initialize $\vartheta$.\\
Cycle
\begin{wideenumerate}
\item $\Sigma_\lambda^q \leftarrow$ $\left\{ {\Sigma_\lambda^0}^{-1} + (F_1+F_2) \mathcal{H}(n,C_\gamma^q,A_\gamma^2) / \mathcal{H}(n-2,C_\gamma^q,A_\gamma^2)  \right\}^{-1}$, where \\ [1mm]
$F_1=\sum_{i=1}^n\sum_{r=1}^m y_i\exp(-\tfrac{1}{2}t_{ir}^T\Sigma_\lambda^qt_{ir})\{{\mu_\alpha^q}_r \cos(t_{ir}^T\mu_\lambda^q)+{\mu_\alpha^q}_{m+r}\sin(t_{ir}^T\mu_\lambda^q)\}t_{ir}t_{ir}^T$, \\ [1mm]
$F_2 =-\frac{1}{4}\sum_{i=1}^n\sum_{r=1}^m\sum_{l=1}^m \Big[\nu_{irl}^- \big\{(A_{rl}+D_{rl})\cos({t_{irl}^-}^T\mu_\lambda^q)
+2B_{rl}\sin({t_{irl}^-}^T\mu_\lambda^q) \big\}t_{irl}^- {t_{irl}^-}^T  \\ [1mm]
+\nu_{irl}^+ \big\{(A_{rl}-D_{rl})\cos({t_{irl}^+}^T\mu_\lambda^q)
+2B_{rl}\sin({t_{irl}^+}^T\mu_\lambda^q) \big\}t_{irl}^+ {t_{irl}^+}^T \Big]$.
\item $\mu_\lambda^q \leftarrow$ $\mu_\lambda^q+\Sigma_\lambda^q \left\{ {\Sigma_\lambda^0}^{-1}(\mu_\lambda^0-\mu_\lambda^q)- \frac{1}{2} (F_3+F_4) \mathcal{H}(n,C_\gamma^q,A_\gamma^2) /  \mathcal{H}(n-2,C_\gamma^q,A_\gamma^2)    \right\}$, where\\ [1mm]
$F_3 = -2\sum_{i=1}^n\sum_{r=1}^m y_i \exp(-\tfrac{1}{2}t_{ir}^T\Sigma_\lambda^qt_{ir})\{{\mu_\alpha^q}_{r+m} \cos(t_{ir}^T\mu_\lambda^q)-{\mu_\alpha^q}_r\sin(t_{ir}^T\mu_\lambda^q)\}t_{ir}$ \\ [1mm]
$F_4 =\frac{1}{2}\sum_{i=1}^n\sum_{r=1}^m\sum_{l=1}^m \Big[\nu_{irl}^- \big\{2B_{rl}\cos({t_{irl}^-}^T\mu_\lambda^q)
-(A_{rl}+D_{rl})\sin({t_{irl}^-}^T\mu_\lambda^q) \big\}t_{irl}^-  \\ [1mm]
+\nu_{irl}^+ \big\{2B_{rl}\cos({t_{irl}^+}^T\mu_\lambda^q)
+(D_{rl}-A_{rl})\sin({t_{irl}^+}^T\mu_\lambda^q) \big\}t_{irl}^+ \Big]$.
\item $\Sigma_\alpha^q \leftarrow \left\{E_q(Z^TZ) \mathcal{H}(n,C_\gamma^q,A_\gamma^2)/ \mathcal{H}(n-2,C_\gamma^q,A_\gamma^2) +m I_{2m} \mathcal{H}(2m,C_\sigma^q,A_\sigma^2) / \mathcal{H}(2m-2,C_\sigma^q,A_\sigma^2)  \right\}^{-1}$ 
\item $\mu_\alpha^q \leftarrow  \Sigma_\alpha^q E_q(Z)^T y \; \mathcal{H}(n,C_\gamma^q,A_\gamma^2) / \mathcal{H}(n-2,C_\gamma^q,A_\gamma^2)$ 
\item $C_\sigma^q \leftarrow \frac{m}{2}\left\{{\mu_\alpha^q}^T\mu_\alpha^q+\text{tr}(\Sigma_\alpha^q)\right\}$ 
\item $C_\gamma^q \leftarrow \frac{1}{2} \left[y^Ty-2y^T E_q(Z)\mu_\alpha^q+\text{tr}\{(\mu_\alpha^q{\mu_\alpha^q}^T+\Sigma_\alpha^q)E_q(Z^TZ)\}\right]$ 
\end{wideenumerate}
until the increase in the lower bound $\mathcal{L}$ is negligible. 
\vspace{1mm}
\hrule}
\caption{Nonconjugate variational message passing algorithm for sparse spectrum GP regression model.}\label{Alg1}
\end{Algorithm*} 

A unique aspect of our variational scheme is the way covariance function uncertainty is handled, with the expectations involving $\lambda$ in the lower bound computable in closed form.  In particular, $E_q(Z)$ and $E_q(Z^TZ)$ can be evaluated in closed form (see Appendix A). Let $\mu_\alpha^q{\mu_\alpha^q}^T +\Sigma_\alpha^q$ be partitioned as $\left[\begin{smallmatrix} A & B^T \\ B & D \end{smallmatrix}\right]$ where $A$, $B$ and $D$ are all $m \times m$ matrices. In algorithm 1, we define $t_{ir}=s_r \odot x_i$,
\begin{equation*}
\begin{aligned}
t_{irl}^- &=t_{ir}-t_{il} \\
t_{irl}^+ &=t_{ir}+t_{il}
\end{aligned} \quad \text{and} \quad
\begin{aligned}
\nu_{irl}^- &= \exp(-\tfrac{1}{2}{t_{irl}^-}^T\Sigma_\lambda^qt_{irl}^-) \\
\nu_{irl}^+ &= \exp(-\tfrac{1}{2}{t_{irl}^+}^T\Sigma_\lambda^qt_{irl}^+)
\end{aligned}
\end{equation*}
for $i=1,\dots,n$, $r=1,\dots,m$, $l=1,\dots,m$. 

The lower bound $\mathcal{L}$ defined in \eqref{LB} is commonly used for monitoring convergence. It can be evaluated in closed form (see Appendix B) and is given by
\begin{equation}\label{LBexp}
\begin{aligned}
\mathcal{L} & =m\log m+\log({4A_\sigma A_\gamma}/{\pi^2})+\tfrac{1}{2}\log|{\Sigma_\lambda^0}^{-1}\Sigma_\lambda^q| \\
&\quad  -\tfrac{1}{2}(\mu_\lambda^q-\mu_\lambda^0)^T{\Sigma_\lambda^0}^{-1} (\mu_\lambda^q-\mu_\lambda^0) -\tfrac{1}{2}\text{tr}({\Sigma_\lambda^0}^{-1}\Sigma_\lambda^q) \\
&\quad +\tfrac{1}{2} \log|\Sigma_\alpha^q|+\log\mathcal{H}(n-2,C_\gamma^q,A_\gamma^2) \\
&\quad +\log\mathcal{H}(2m-2,C_\sigma^q,A_\sigma^2)+m+ \tfrac{d}{2}-\tfrac{n}{2}\log(2\pi).  
\end{aligned}
\end{equation}
The above expression applies only after the updates in steps 5 and 6 of Algorithm \ref{Alg1} have been made.

\section{Adaptive nonconjugate variational message passing} \label{adaptive strategy}
In the sparse spectrum GP regression model \eqref{model2}, $Z$ and $\alpha$ are intimately linked. Each time the lengthscales $(\lambda)$ are changed by a small amount, the amplitudes ($\alpha$) will have to respond to this change in order to match the observed $y$. In \eqref{pdt}, we have assumed that the variational posteriors of $\lambda$ and $\alpha$ are independent so that expectations with respect to $q$ are tractable and closed form updates can be derived for a fast algorithm. However, strong dependence between $\lambda$ and $\alpha$ implies that only small steps can be taken in each cycle of updates and a large number of iterations will likely be required for Algorithm \ref{Alg1} to converge. 

To accelerate convergence, we propose modifying the updates in steps 1 and 2. Let $\eta_\lambda$ be the natural parameter of $q(\lambda)$ and $\hat{\eta}_\lambda$ be the update of $\eta_\lambda$ in nonconjugate variational message passing. \cite{Tan2014} showed that nonconjugate variational message passing is a natural gradient ascent method with step size one. At iteration $t$, we consider 
\begin{equation}\label{step size update}
\begin{aligned}
\eta_\lambda^{(t)} &= \eta_\lambda^{(t-1)} + a_t \tilde{\nabla}_{\eta_\lambda}\mathcal{L}|_{\eta_\lambda=\eta_\lambda^{(t-1)}} \\
&=\eta_\lambda^{(t-1)} + a_t \left(\hat{\eta}^{(t)} _\lambda- \eta_\lambda^{(t-1)} \right)\;\; (\text{from}\;\; \eqref{natgrad})
\end{aligned}
\end{equation}
where $\hat{\eta}_\lambda^{(t)} = \mathcal{V}_\lambda(\eta_\lambda^{(t-1)})^{-1}\sum_{a \in N(\lambda)} \frac{\partial S_a}{\partial \eta_\lambda} \big|_{\eta_\lambda^{(t-1)}}$. When $a_t=1$, \eqref{step size update} reduces to the update in nonconjugate variational message passing. Taking $a_t<1$ may be helpful when updates in nonconjugate variational message passing fail to increase $\mathcal{L}$. From our experiments, instability in Algorithm 1 usually occur within the first few iterations. Beyond that, the algorithm is usually quite stable and taking larger steps with $a_t>1$ can result in significant speed-ups.

\begin{figure}
\centering
\includegraphics[width=0.49\textwidth]{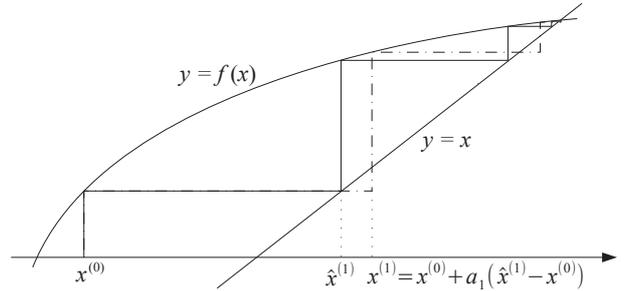} 
\caption{Solid line starting from $x^{(0)}$ indicates conventional path in fixed point iterations while the dot dash line indicates path to convergence with a step size greater than 1.}
\label{fixedpointplot}
\end{figure}

Recall that nonconjugate variational message passing is a fixed point iterations algorithm. Figure \ref{fixedpointplot} illustrates in a single variable case (where we are solving $x=f(x)$) how taking steps larger than one can accelerate convergence. Instead of taking $x^{(t)}=f( x^{(t-1)})$, consider $x^{(t)}=x^{(t-1)}+a_t(\hat{x}^{(t)}-x^{(t-1)})$, where $\hat{x}^{(t)}=f( x^{(t-1)})$ and $a_t>1$. The solid line starting from $x^{(0)}$ indicates the conventional path in fixed point iterations while the dot dash line indicates the path with a step size greater than 1. The dot dash line moves towards the point of convergence faster than the solid line. However, it may overshoot if $a_t$ is too large. In Algorithm \ref{Alg2}, we borrow ideas from \cite{Salak2003} to construct an adaptive algorithm where $a_t$ is allowed to increase by a factor $\rho>1$ after each cycle of updates whilst $\mathcal{L}$ is on an increasing trend and we revert to $a_t=1$ when $\mathcal{L}$ decreases. 

The adaptive nonconjugate variational message passing algorithm is given in Algorithm \ref{Alg2}. 
\begin{Algorithm}
\centering
\parbox{0.48\textwidth}{
\hrule
\vspace{1mm}
Initialize $\vartheta^{(0)}$. Set $t=0$ and $a_0=1$. \\
While $\delta > \text{tolerance}$ and $t< \text{maximum number of iterations}$,
\begin{wideenumerate}
\item $t \leftarrow t+1$.
\item Compute $F_5 = {\Sigma_\lambda^0}^{-1} + \frac{\mathcal{H}(n,C_\gamma^q,A_\gamma^2)}{\mathcal{H}(n-2,C_\gamma^q,A_\gamma^2)}(F_1+F_2)$ and \\ $F_6 =  {\Sigma_\lambda^0}^{-1}(\mu_\lambda^0-\mu_\lambda^q)-\frac{\mathcal{H}(n,C_\gamma^q,A_\gamma^2)}{2\mathcal{H}(n-2,C_\gamma^q,A_\gamma^2)}(F_3+F_4)$.
\item 
\begin{wideenumerate}
\item Compute  $\Sigma_\lambda^q \leftarrow \left[ (1-a_t) {\Sigma_\lambda^q}^{-1} + a_t F_5 \right]^{-1} $. 
\item If $\Sigma_\lambda^q$ is symmetric positive definite, proceed to step 4.
Else, $a_t \leftarrow a_t/\rho$ and return to step 3(a).
\end{wideenumerate}
\item $\mu_\lambda^q \leftarrow$ $\mu_\lambda^q+a_t\Sigma_\lambda^q F_6$. 
\item Compute updates in steps 3--6 of Algorithm \ref{Alg1}. 
\item 
\begin{wideenumerate}
\item Compute $\delta=\mathcal{L}|_{\vartheta^{(t)}}-\mathcal{L}|_{\vartheta^{(t-1)}}$.
\item If $\delta>0$, $a_t=\rho \,a_{t-1}$  and return to step 1. Else, $a_t \leftarrow 1$, $t \leftarrow t+1 $ and return to step 3.
\end{wideenumerate}
\end{wideenumerate}
\hrule}
\caption{Adaptive nonconjugate variational message passing algorithm for sparse spectrum GP regression model.}\label{Alg2}
\end{Algorithm} 
In Appendix C, we show that \eqref{step size update} reduces to the updates:
\begin{multline}\label{simplified adaptive updates}
\Sigma_\lambda^q  \leftarrow \bigg[ (1-a_t) {\Sigma_\lambda^q}^{-1} -2a_t \text{vec}^{-1}\bigg(\sum_{a \in N(\lambda)} \frac{\partial S_a}{\partial  \text{vec}(\Sigma_\lambda^q)}\bigg) \bigg]^{-1}  \\
\text{and} \;\; \mu_\lambda^q \leftarrow \mu_\lambda^q + a_t\,  \Sigma_\lambda^q \sum_{a \in N(\lambda)} \frac{\partial S_a}{\partial  \mu_\lambda^q}. 
\end{multline}
Step 3(b) has been added as a safeguard as the updated $\Sigma_\lambda^q$ may not be symmetric positive definite due to rounding errors or when $a_t$ is large. In this case, we propose reducing the step size by a factor $\rho$ until all eigenvalues of $\Sigma_\lambda^q$ are positive. It is useful to insert step 3(b) in Algorithm \ref{Alg1} after $\Sigma_\lambda^q$ has been updated as well as it can serve as damping. For both Algorithms \ref{Alg1} and \ref{Alg2}, we initialize $\mu_\lambda^q$ as $[0.5,\dots,0.5]^T$ (which is one half of the amplitudes of the inputs after any rescaling), $\Sigma_\lambda^q$ as $\text{diag}[0.5,\dots,0.5]^T$, $C_\gamma^q$ as $(\frac{n}{2}-1) \cdot \text{var}(y)/4$, $C_\sigma^q$ as $(m-1) \cdot \text{var}(y)$, and $\mu_\alpha^q$ and $\Sigma_\alpha^q$ are initialized using the updates in steps 3--4 of Algorithm \ref{Alg1}. We set the maximum number of iterations as 500 and the algorithms are deemed to have converged if the relative increase in $\mathcal{L}$ is less than $10^{-6}$. \cite{Salak2003} recommend taking the factor $\rho$ to be close to but more than 1. Using this as a guide, we have experimented with $\rho$ taking values 1.1, 1.5 and 2. While all these values lead to improvement in efficiency, we find $\rho=1.5$ to be more favourable, as the step sizes increase rather slowly when $\rho=1.1$ and too fast when $\rho=2$, leading to many failed attempts to improve $\mathcal{L}$. While Algorithm \ref{Alg2} does not necessarily converge to the same local mode as Algorithm \ref{Alg1}, results from the two algorithms are usually very close. Algorithm \ref{Alg2} sometimes demonstrates the ability to avoid local modes with the larger steps that it takes. We compare and quantify the performance of the two algorithms in Section \ref{pendulum eg}. Note that in Algorithm 2, each failed attempt to improve $\mathcal{L}$ is also counted as an additional iteration in step 5(b) even though step 1 does not have to be reevaluated. 

We note that Algorithms 1 and 2 are not guaranteed to converge due to the fixed point updates in nonconjugate variational message passing. However, convergence issues can usually be mitigated by rescaling variables and varying the initialization values. As the fixed point updates may not result in an increase in $\mathcal{L}$, it is possible to compute $\mathcal{L}$ after performing the updates and reduce $a_t$ if necessary. However, this requires computing a lower bound of a more complex form than \eqref{LBexp} at each iteration. Our experiments indicate that a decline in $\mathcal{L}$ is often due to $\Sigma_\lambda^q$ not being symmetric positive definite, and hence installing step 3(b) suffices in most cases. We also find that checking the simplified form of $\mathcal{L}$ in \eqref{LBexp} at the end of each cycle and simply reverting $a_t$ to 1 if necessary is more economical. If premature stopping occurs in Algorithms 1 or 2 due to a decrease in the lower bound at some iteration, this can be detected by examination of the lower bound values and remedied if needed by damping where values $a_t<1$ are considered.  

\section{Predictive distribution and performance evaluation} \label{pred distn}
Let $D =\{(x_i,y_i)|i=1,\dots,n\}$ and $T =\{(x_j^*,y_j^*)|j=1,\dots,n^*\}$ be the training and testing data sets respectively. Let $S=\{s_1,\dots,s_m\}$ be the set of spectral frequencies randomly generated from $N(0,I_d)$. Bayesian predictive inference is based on the predictive distribution, 
\begin{multline*}
p(y_j^*|x_j^*,S,D)
=\int p(y_j^*|x_j^*,S,\alpha,\lambda,\gamma) \\ \cdot p(\alpha,\lambda,\gamma|D,S) \; d\alpha\;d\lambda\;d\gamma,
\end{multline*}
assuming $y_j^*$ is conditionally independent of $D$ given $\alpha$, $\lambda$ and $\gamma$. We replace $p(\alpha,\lambda,\gamma|D)$ with our variational approximation $q(\alpha,\lambda,\gamma)=q(\alpha)q(\lambda)q(\gamma)$ so that 
\begin{multline}\label{postpred}
p(y_j^*|x_j^*,S,D) \approx \int p(y_j^*|x_j^*,S,\alpha,\lambda,\gamma) \\ \cdot q(\alpha)q(\lambda)q(\gamma) \; d\alpha\;d\lambda\;d\gamma. 
\end{multline}
From (\ref{postpred}), the posterior predictive mean of $y_j^*$ is 
\begin{multline*}
\mu_j^* = \int y_j^* \; p(y_j^*|x_j^*,S,D)  \; dy_j^* \\
\approx E_q\left\{ \int y_j^*\; p(y_j^*|x_j^*,S,\alpha,\lambda,\gamma)  \; dy_j^* \right \} =E_q(Z_j^*)^T \mu_\alpha^q,
\end{multline*}
where 
\begin{multline*}
Z_j^*=[
\cos\{(s_1\odot x_j^*)^T\lambda \},\dots,\cos\{(s_m\odot  x_j^*)^T\lambda \}, \\
\sin\{(s_1\odot  x_j^*)^T\lambda \},\dots,\sin\{(s_m\odot  x_j^*)^T\lambda \}]^T
\end{multline*}
and $E_q(Z_j^*)$ can be computed using results in Appendix A. The posterior predictive variance is 
\begin{equation*}
\begin{aligned}
{\sigma_j^*}^2 &\approx \int {y_j^*}^2 \; p(y_j^*|x_j^*,S,D)  \; dy_j^* - \{E_q(Z_j^*)^T \mu_\alpha^q \}^2  \\
&\approx E_q\{ \gamma^2+ (Z_j^*\alpha)^2 \} -{\mu_\alpha^q}^TE_q(Z_j^*)E_q(Z_j^*)^T\mu_\alpha^q \\
&=\tfrac{\mathcal{H}(n-4,C_\gamma^q,A_\gamma^2)}{\mathcal{H}(n-2,C_\gamma^q,A_\gamma^2)}+\text{tr}\{(\mu_\alpha^q{\mu_\alpha^q}^T+\Sigma_\alpha^q)E_q({Z_j^*}^TZ_j^*)\} \\
& \quad -{\mu_\alpha^q}^TE_q(Z_j^*)E_q(Z_j^*)^T\mu_\alpha^q.
\end{aligned}
\end{equation*}

In the examples, we follow \citet{Lazaro2010} and evaluate performance using two quantitative measures: normalized mean square error (NMSE) and mean negative log probability (MNLP). These are defined as 
\begin{equation*}
\begin{aligned}
\text{NMSE} &=\frac{\frac{1}{n^*}\sum_{j=1}^{n^*}(y_j^*-\mu_j^*)^2}{\frac{1}{n^*}\sum_{j=1}^{n^*}(y_j^*-\bar{y})^2} \;\; \text{and}\;\; \\
\text{MNLP} &=\frac{1}{2n^*}\sum_{j=1}^{n^*}\left\{ \frac{(y_j^*-\mu_j^*)^2}{{\sigma_j^*}^2}+\log {\sigma_j^*}^2 +\log(2\pi) \right\}.
\end{aligned}
\end{equation*}
The MNLP is implicitly based on a normal predictive distribution for $y_j^*$ with mean $\mu_j^*$ and variance ${\sigma_j^*}^2$, $j=1,\dots,n^*$.

\section{Adaptive neighbourhoods approach for predictive inference} \label{neigh}

We propose a new technique of obtaining predictive inference by fitting models locally using adaptive neighbourhoods. Our proposed approach consists of two stages:
For each test point $x_j^*$, $j=1,\dots,n^*$,
\begin{enumerate}
\item we first find the $k$ nearest neighbours of $x_j^*$ in $D$ (that are closest to $x_j^*$ in terms of Euclidean distance) and denote the index set of these $k$ neighbours by ${N}_1$. We use Algorithm 2 to fit a sparse spectrum GP regression model, $M_1$, to $\{(x_i,y_i)|i \in {N}_1\}$. 
\item Next, we use the variational posterior mean of the lengthscales, $\mu_\lambda^q$, from $M_1$ to define a new distance measure:
\begin{equation}\label{lengthscale_dist}
d(x_j^*,x_i)= \sqrt{(x_j^*-x_i)^T \text{diag}({\mu_\lambda^q}^2) (x_j^*-x_i)},
\end{equation}
where the dimensions are weighted according to ${\mu_\lambda^q}^2$. This will effectively downweight or remove variables of little or no relevance. Using this new distance measure, we find the $k$ nearest neighbours of $x_j^*$ in $D$ and denote the index set of these $k$ neighbours by ${N}_2$. We use Algorithm 2 to fit a sparse spectrum GP regression model, $M_2$, to $\{(x_i,y_i)|i \in {N}_2\}$ and use the variational posterior from $M_2$ for predictive inference.
\end{enumerate}
In summary, the first fitting $(M_1)$ is used to find out which variables are more relevant in determining the output. From \eqref{sqexpcov}, a large value of $\lambda_l$ indicates that the covariance drops rapidly along the dimension of $l$ and hence the neighbourhood should be shrunk along the $l$th dimension. Using $\mu_\lambda^q$ from the first fit as an estimate of the lengthscales, the neighbourhood is then adapted before performing a second fitting $(M_2)$ to improve prediction. We do not recommend iterating the fitting process further since this may result in cyclical behaviour with the neighbourhood successively expanding and contracting along a certain dimension as the iterations proceed.  In the examples, when the SSGP algorithm is implemented using this adaptive neighbourhood approach, we replace the variational posterior mean value $\mu_\lambda^q$ (which does not exist for the SSGP method since it does not estimate a variational posterior distribution for $\lambda$) by the point estimates of the lengthscales $\hat{\lambda}$ obtained by the SSGP approach.

The adaptive neighbourhood approach is well-placed to handle data with nonstationarities as stationarity is only assumed locally and local fitting can adapt the noise and the degree of smoothing to the nonstationarities. Adapting the neighbourhood can also be very helpful in improving prediction when there are many irrelevant variables due to automatic relevance determination implemented via the lengthscales. A major advantage of the variational approach is that it allows uncertainty in the covariance hyperparameters to be modelled within a fast computational scheme. This is especially important when fitting using local neighbourhoods as plug-in approaches to estimating hyperparameters will tend to underestimate predictive uncertainty when the data set is small. This approach is advantageous for dealing with large data sets as well. As we only consider fitting models to a small subset $k$ of data points at each test point, a smaller number of basis functions $(m)$ might suffice. While the computational requirements grow linearly with the number of prediction locations, this approach is trivially parallelizable to get a linear speed-up with the number of processors.

\section{Examples} \label{eg}

We compare the performance of the variational approach with the SSGP algorithm using three real data sets: the pendulum data set, the rainfall-runoff data set and the Auto-MPG data set. The implementation of SSGP in Matlab is obtained from \url{http://www.tsc.uc3m.es/~miguel/simpletutorialssgp.php}. There are two versions of the SSGP algorithm: SSGP (fixed) uses fixed spectral points while SSGP (optimized) optimizes the marginal likelihood with respect to the spectral points. We will only consider SSGP (fixed). We observe some sensitivity in predictive performance to the basis functions and adopt the following strategy for better results: for each implementation of Algorithm \ref{Alg1} (or \ref{Alg2}), we randomly generate ten sets of spectral points from $N(0,I_d)$, perform 2 iterations of the algorithm, and select the set with the highest attained lower bound to continue to full convergence. A similar strategy was used by \cite{Lazaro2010} to initialize the SSGP algorithm. Due to the zero mean assumption, we center all target vectors, $y$ by subtracting the mean $\bar{y}$ from $y$. In the examples, ``VA'' refers to the variational approximation approach implemented via Algorithm \ref{Alg2}, ``global''  refers to using the entire training set for fitting while ``local'' refers to the adaptive neighbourhood approach described in Section \ref{neigh}. 

\subsection{Pendulum Data Set} \label{pendulum eg}

The pendulum data set (available at \url{http://www.tsc.uc3m.es/~miguel/simpletutorialssgp.php}) has $d=9$ covariates and contains 315 training points and 315 test points. The target variable is the change in angular velocity of a simulated mechanical pendulum over 50 ms and the covariates consist of different parameters of the system. \cite{Lazaro2010} used this example to show that SSGP (optimized) can sometimes fail due to overfitting. We rescale the input variables in the training set to lie in $[-1,1]$ and consider the number of basis functions, $m \in \{10,25,50,100,200\}$. We compare the performance of Algorithm \ref{Alg2} with SSGP (fixed) using NMSE and MNLP values averaged over ten repetitions. We set $\rho=1.5$, $A_\sigma=A_\gamma=25$ for the half-Cauchy priors, following \cite{Gelman2006} and \cite{Wand2011} and $\mu_\lambda^0=0$, $\Sigma_\lambda^0=10I_d$ for the lengthscales in Algorithm \ref{Alg2}. 

For this data set which is quite small, we note that the adaptive neighbourhood approach did not yield significant improvements as all inputs are relevant and there is no strong nonstationarity. Hence we report only results for global fits, which are shown in Figure \ref{pendulumNMSEMNLP}. The NMSE and MNLP values produced by Algorithm \ref{Alg2} are comparable with that of SSGP (fixed) for small $m$ and are better for large $m$. On the whole, Algorithm \ref{Alg2} produces reasonably good NMSE performance and is less prone to overfitting than the SSGP algorithm. The ability of the variational approach to treat uncertainty in the covariance function hyperparameters reduces underestimation of predictive uncertainty, resulting in better MNLP performance.

\begin{figure}
\centering
\includegraphics[width=0.495\textwidth]{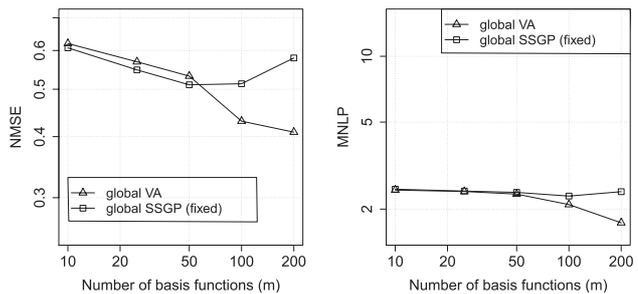} 
\caption{Pendulum data set. NMSE (left) and MNLP (right) values produced by Algorithm \ref{Alg2} (global VA) and global SSGP (fixed) and averaged over ten repetitions plotted against number of basis functions $(m)$.}
\label{pendulumNMSEMNLP} 
\end{figure}

\begin{figure}
\centering
\includegraphics[width=0.495\textwidth]{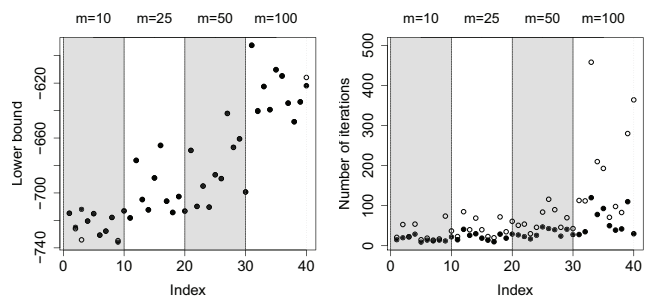} 
\caption{Pendulum data set. Left: Plot of lower bound attained at convergence against index of runs. Right: Plot of number of iterations required for convergence against index of runs. Solid circles correspond to Algorithm \ref{Alg2} with $\rho=1.5$ while empty circles correspond to Algorithm \ref{Alg1}.}
\label{pendulumLBcomparisons} 
\end{figure}
Next, we compare the performance of Algorithm \ref{Alg1} with Algorithm \ref{Alg2} both in terms of efficiency and the lower bound attained at convergence. We use Algorithm \ref{Alg1} to re-perform the runs for $m \in \{10,25,50,100\}$, using the same sets of spectral points that were used in Algorithm 2. These runs are indexed from 1 to 40 (there are ten repetitions for each $m$). Figure \ref{pendulumLBcomparisons} shows a plot of the lower bound attained at convergence on the left and a plot of the number of iterations required for convergence on the right for each of the 40 runs. Figure \ref{pendulumLBcomparisons} indicates that, except for runs 3 and 40, the lower bound attained by Algorithms \ref{Alg1} and \ref{Alg2} are almost indistinguishable. However, Algorithm \ref{Alg2} required a much smaller number of iterations to converge than Algorithm \ref{Alg1}. Excluding runs 3 and 40 where the lower bound attained by Algorithms \ref{Alg1} and \ref{Alg2} differs significantly, using Algorithm \ref{Alg2} instead of Algorithm \ref{Alg1} leads on average to a reduction of 49\% in the number of iterations required for convergence. The highest reduction observed is 84\% at run 9. At run 3, Algorithm \ref{Alg2} was able to escape a local mode and attained a higher lower bound at convergence. However, at run 40, it was caught in a local mode. We re-perform run 40 using $\rho=1.1$ and it turns out that Algorithm \ref{Alg2} was then able to attain the same lower bound as Algorithm \ref{Alg1} but in around half the number of iterations.

\begin{figure}
\centering
\includegraphics[width=0.495\textwidth]{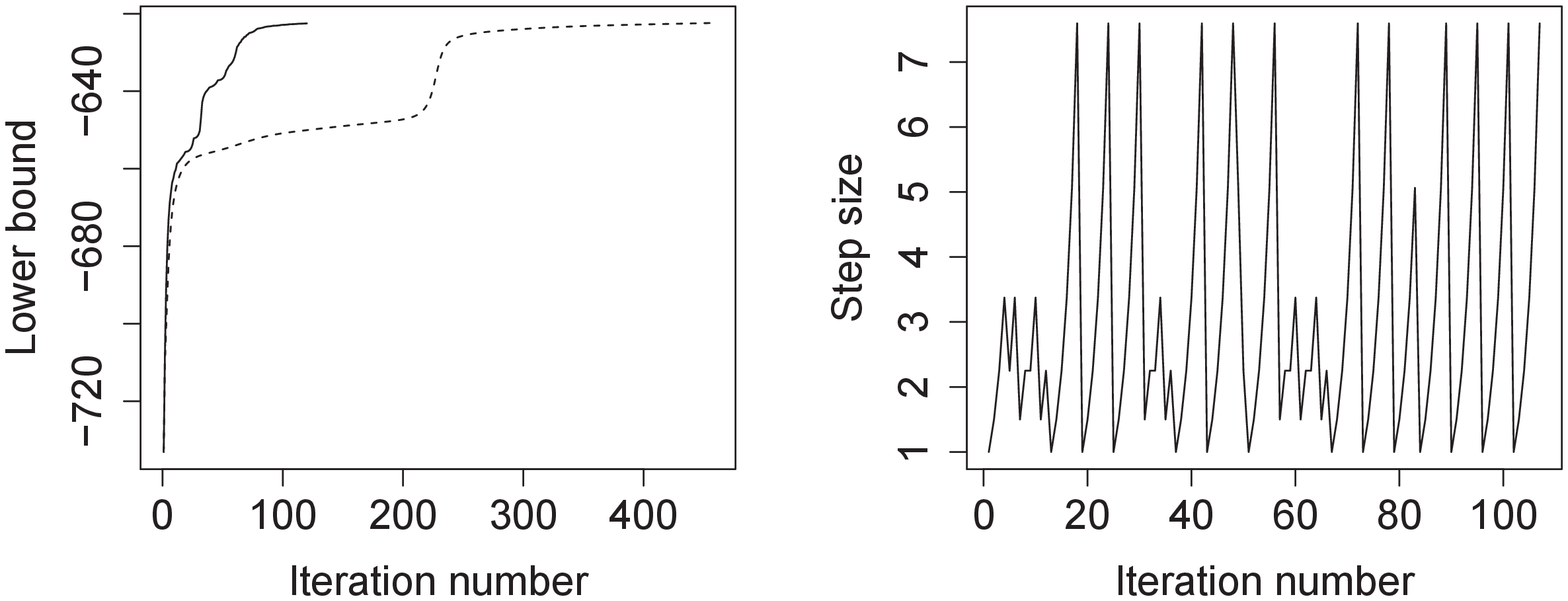} 
\caption{Pendulum data set. Run 33. Left: Plot of lower bound against iteration number (solid line corresponds to Algorithm \ref{Alg2} with $\rho=1.5$ while dashed line corresponds to Algorithm \ref{Alg1}). Right: Plot of the adaptive step size $(a_t)$ used in Algorithm \ref{Alg2} against iteration number $(t)$.}
\label{pendulumrun33} 
\end{figure}
The typical behaviour of Algorithm \ref{Alg2} is illustrated in Figure \ref{pendulumrun33}. On the left is a plot of the lower bound against iteration number and on the right is a plot of the adaptive step size $(a_t)$ used in Algorithm \ref{Alg2} against iteration number $(t)$ for run 33. The step size typically increases by a factor of 1.5 at each iteration but falls back to 1 when the lower bound fails to increase. The step size may also be reduced by factors of 1.5 due to the requirement that the covariance matrix be symmetric positive definite in step 2(b) of Algorithm \ref{Alg2}. The reduction in the number of iterations that Algorithm \ref{Alg2} takes to converge as compared to Algorithm \ref{Alg1} is 74\% for run 33.

\subsection{Performance of SSGP and VA with adaptive neighbourhood approach}
 
For the next four subsections, our discussion concerns performance of the adaptive neighbourhood approach. We compare the performance using two real datasets: the rainfall-runoff data set and the Auto-MPG data set. We fit these data globally using SSGP (fixed), VA and MCMC, and compare results with the adaptive neighbourhood approach, implemented using both SSGP (fixed) and Algorithm \ref{Alg2} with factor $\rho=1.5$. For the priors, we set $A_\sigma=A_\gamma=25$, $\mu_\lambda^0=0$. For $\Sigma_\lambda^0$, we set $\Sigma_\lambda^0 = 100I_d$ for the rainfall-runoff data where a less smooth mean function is expected and $\Sigma_\lambda^0 = I_d$ for the Auto-MPG data. The prior variance for the lengthscales can be chosen empirically by predictive performance on a test set or using prior knowledge. Prior knowledge about the hyperparameters in the covariance function can be elicited by thinking about the prior degree of expected correlation of the mean function for covariates separated by lag one in each dimension when the covariates are standardized. For both the global SSGP (fixed) and VA approach, we consider the number of basis functions $m \in \{20,40,60,80,100\}$. In addition, we generate ten artificial covariates on top of the existing covariates in both the rainfall-runoff and Auto-MPG data set to test the capability of Algorithm \ref{Alg2} in automatic relevance determination.

\subsection{Rainfall-runoff data}  \label{rainfall}

In this example, we consider data from a deterministic rainfall-runoff model, which is a simplification of the Australian Water Balance Model \citep[AWBM,][]{Boughton2004}. The AWBM estimates catchment streamflow using time series of rainfall and evapotranspiration data and is widely used in Australia for estimating catchment water yield or design flood estimation. The model has three parameters - the maximum storage capacity $S$, the base flow index \textit{BFI} and the baseflow recession factor $K$. We have model simulations for around eleven years of average monthly potential evapotranspiration and daily rainfall data for the Barrington River catchment, located in New South Wales, Australia. The model was run for 500 different values of the parameters $(S,K,BFI)$ generated using a maximin Latin hypercube design. This data contains 500 data points for each of 3700 days, with a total of 1.85 million data points. For each day, the total rainfall is also recorded. A subset of this data has been studied in \cite{Nott2012}. 

\begin{figure*} [htb!]
\centering
\includegraphics[width=0.66\textwidth, angle=270]{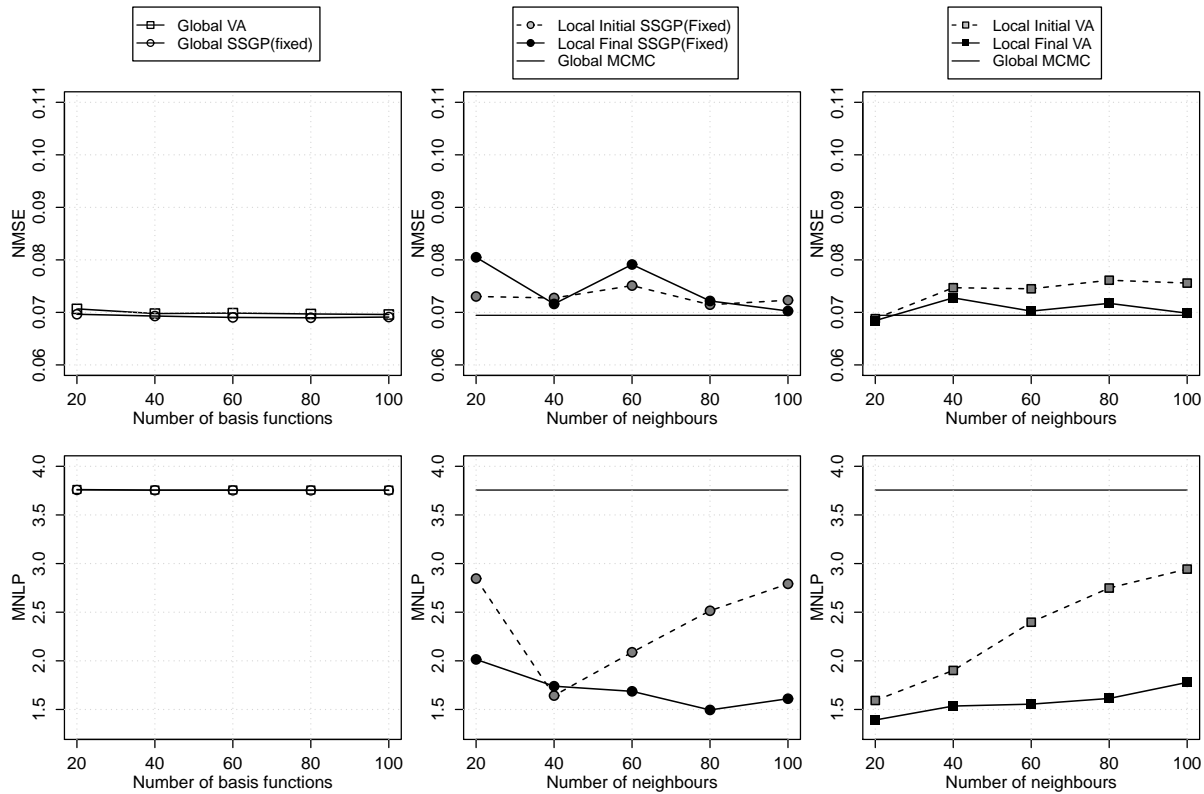} 
\caption{\label{rainfall2covMNLPNMSE} Rainfall-runoff data on day with peak rainfall. NMSE and MNLP values averaged over ten repetitions plotted against the number of basis functions (first column) and against the number of neighbours (second and third columns). Number of basis functions used in the local methods was 20.}
\end{figure*}

Even though the size of the data is large, the computational demands of the adaptive neighbourhood approach will depend mostly on the number of query points and the neighbourhood size. This makes our approach highly suitable for this data set. This is especially true since emulation of the model will be most interesting near values of peak rainfall input and generally for events of hydrological significance, where there might be a flood risk for example.  So the proportion of interesting query points in this example is a small fraction of the total data set size and furthermore we expect the model output to vary rapidly in some parts of the parameter space but very little in other parts so the ability of the local method to smooth adaptively is very attractive for this problem. We will consider prediction for the two days with the highest rainfall inputs. We take AWBM streamflow response as the target $y$, and $S$ and $K$ as covariates, omitting \textit{BFI}. A small amount of independent normal random noise with standard deviation 0.01 was added to $y$ to avoid degeneracies in regions of the space where the response tends to be identically zero. For each day, we randomly selected 100 data points as the test set and use the remaining 400 data points as the training set. These data are highly nonstationary with large flat regions, a few rapidly varying regions and the noise level changes a lot over the space. 

Figure \ref{rainfall2covMNLPNMSE} shows the NMSE and MNLP values averaged over ten repetitions for the rainfall-runoff data with peak rainfall. For global SSGP (fixed), we observe a slight improvement in NMSE values as $m$ increases, while MNLP values remain largely constant at around 3.75 even for large $m$. Due to the nonstationary nature of this data, a global stationary fit does very poorly in MNLP. For the adaptive neighbourhood approach, we consider neighbourhoods of size $k=20, 40, 60, 80, 100$, fixing the number of basis functions, $m=20$. For the local methods, the dotted lines correspond to results from the initial fitting where the $k$ nearest neighbours are determined based on Euclidean distance. The solid lines correspond to results from the final fit where the $k$ nearest neighbours are determined using the new distance measure with dimensions weighted according to the lengthscales. The improvement brought about by adapting the neighbourhood is more apparent in VA than in SSGP (fixed). 

\begin{figure*} [htb!]
\centering
\includegraphics[width=0.66\textwidth, angle=270]{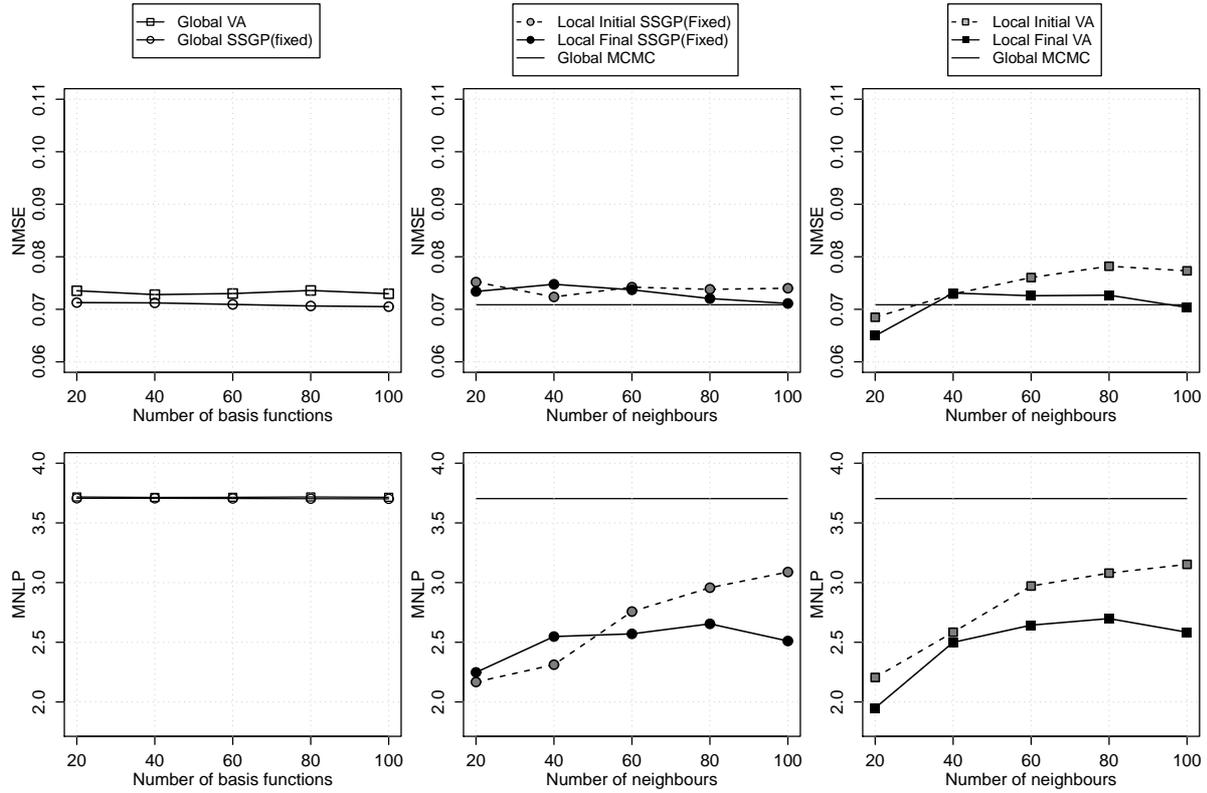} 
\caption{\label{rainfall2covMNLPNMSE_2} Rainfall-runoff data on day with second highest rainfall. NMSE and MNLP values averaged over ten repetitions plotted against the number of basis functions (first column) and against the number of neighbours (second and third columns). Number of basis functions used in the local methods was 20.}
\end{figure*}

Figure \ref{rainfall2covMNLPNMSE_2} shows the NMSE and MNLP values averaged over ten repetitions for the rainfall-runoff data with the second highest rainfall. In this example we also observe that a global stationary fit does very poorly in MNLP, again due to the nonstationary nature of the data. Similarly, when adapting the neighbourhood approach, there are greater improvements in VA than in SSGP (fixed). It is clear that the adaptive neighbourhood approach is critical for this data set where the mean function varies rapidly over some parts of the space but very little over other parts. The variational approach performs very well when using just a small neighbourhood about each test point both in terms of NMSE and MNLP.

\begin{figure}[htb!]
\centering
\includegraphics[width=0.495\textwidth]{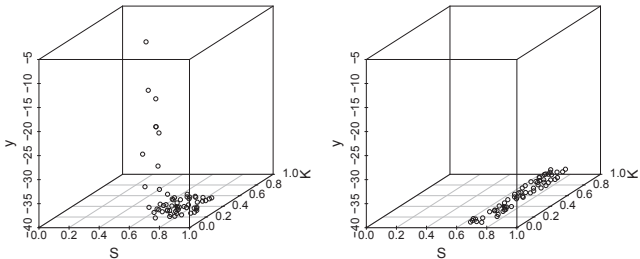} 
\caption{Rainfall-runoff data on day with peak rainfall. Plot of neighbourhood of test point determined using Euclidean distance (left) and new weighted distance measure (right). Circles denote neighbours and solid circle denotes test point.}
\label{s=68} 
\end{figure}
Figure \ref{s=68} illustrates how the neighbourhood of a test point changes from the initial to the final fit for the case $k=60$, when Algorithm \ref{Alg2} was being used. The plot on the left shows the neighbours (denoted by circles) of a test point (denoted by solid circle) determined using Euclidean distance. The plot on the right shows the neighbours of the same test point determined using the new distance measure. In this case, the component of $\mu_\lambda^q$ corresponding to the covariate $S$ is much larger than that corresponding to the covariate $K$, resulting in the neighbourhood being shrunk along the $S$ axis. The adapted neighbourhood leads to an improvement in the estimation of the predictive mean and especially the predictive variance of the test point.

\begin{table}[htb!]
\centering
  \begin{tabular}{|| c | c | c | c ||   } \hline
\multicolumn{4}{||c||}{Global approach}    \\ \hline \\[-2.5ex]
  m & VA & SSGP & MCMC (2000 iterations)  \\ \hline \\[-2.5ex]
  20 & 2.626 & 0.177 &  \\
  40 & 8.315 &  0.268 & \\
  60& 17.545 &  0.491 &  2068.596 \\
  80 & 32.920 &  0.638  &  \\
     100& 55.223 & 0.836 &  \\ \hline
    \multicolumn{4}{||c||}{Adaptive neighbourhood approach $(m=20)$} \\ \hline \\[-2.5ex]
  k & VA & SSGP & MCMC (2000 iterations) \\ \hline \\[-2.5ex]
20 & 257.773 & 16.361 &   653.690 \\
40 & 247.744 &  17.964 &  2253.453 \\
 60 & 265.194 & 22.695 & 5182.298 \\
80 & 320.416 & 23.091 & 9229.747 \\
100 & 279.356 & 24.774& 14441.828 \\ \hline
  \end{tabular}
  \caption{Computation times in seconds for VA, SSGP (fixed) and MCMC for rainfall-runoff data on day with peak rainfall.}
\label{Comtimes}
\end{table}
Table \ref{Comtimes} shows the computation times of the VA, SSGP (fixed) and MCMC algorithms on the rainfall-runoff data with peak rainfall input. We ran the MCMC using Rstan \citep{rstan-software:2014} on a dual processor Windows PC 3.30 GHz workstation and both SSGP (fixed) and VA in Matlab using a 3.2 GHz Intel Core I5 Quad Core iMac. For the global approach, computation times for SSGP (fixed) and VA are averaged over 10 repetitions, while MCMC is based on a single run with 2000 iterations.  For the adaptive neighbourhood approach, Table \ref{Comtimes} shows the total time it takes to run all 100 test points. Note that the timing for the local approaches can be significantly reduced by parallelizing. In terms of computation speed, SSGP (fixed) is the fastest followed by VA. We observe that MCMC is substantially slower than the other methods and the computation time increases significantly when the size of neighbourhood increases. We do not observe such significant increase in computation times for VA and SSGP (fixed).

\subsection{Rainfall-runoff with simulated data}\label{Rain_Sim}
We consider rainfall-runoff data on the day with peak rainfall and generate ten additional covariates artificially. As both covariates $S$ and $K$ lie in the interval $[0,1]$, we simulate each of the ten additional covariates randomly from the uniform distribution on the interval $[0,1]$. We compare the performance of SSGP (fixed)  and Algorithm \ref{Alg2} using a global fit with the adaptive neighbourhood approach. We set $\rho=1.5$ in Algorithm \ref{Alg2} and use the same priors as in Section \ref{rainfall}. For the global approach, we consider the number of basis functions, $m \in \{20,40,60,80,100\}$ while for the adaptive local neighbourhood approach, we consider neighbourhoods of size $k=20, 40, 60, 80, 100$, fixing the number of basis functions, $m=20$. The results are shown in Figure \ref{rainfall12covMNLPNMSE}. 

For the global approach, the results of SSGP (fixed) are quite similar to those in the 2 covariates case.  For the local approach, a small neighbourhood with $k=20$ does not work well for both SSGP (fixed) and VA, indicating that a larger neighbourhood is likely required for high dimensional problems. There is a clear improvement in the MNLP values from adapting the neighbourhood according to the lengthscales and Algorithm 2 achieved the lowest MNLP values among the methods that were studied, using a smaller neighbourhood. The MNLP values achieved by Algorithm \ref{Alg2} are close to those attained in Section \ref{rainfall}, indicating that the adaptive neighbourhood approach is effective in eliminating covariates of little relevance. The variational approach is able to provide significant improvement in this aspect and is much more robust to overfitting for small neighbourhoods. However, the NMSE values obtained in the adaptive neighbourhood approach are higher than those obtained in the global approach. Finally, we note that a good neighbourhood size is dependent on the number of covariance function parameters to be estimated and on the degree of nonstationarity, which is very much problem specific. Some experimentation with different neighbourhood sizes is probably necessary. 

\begin{figure*}[htb!]
\centering
\includegraphics[width=0.66\textwidth, angle=270]{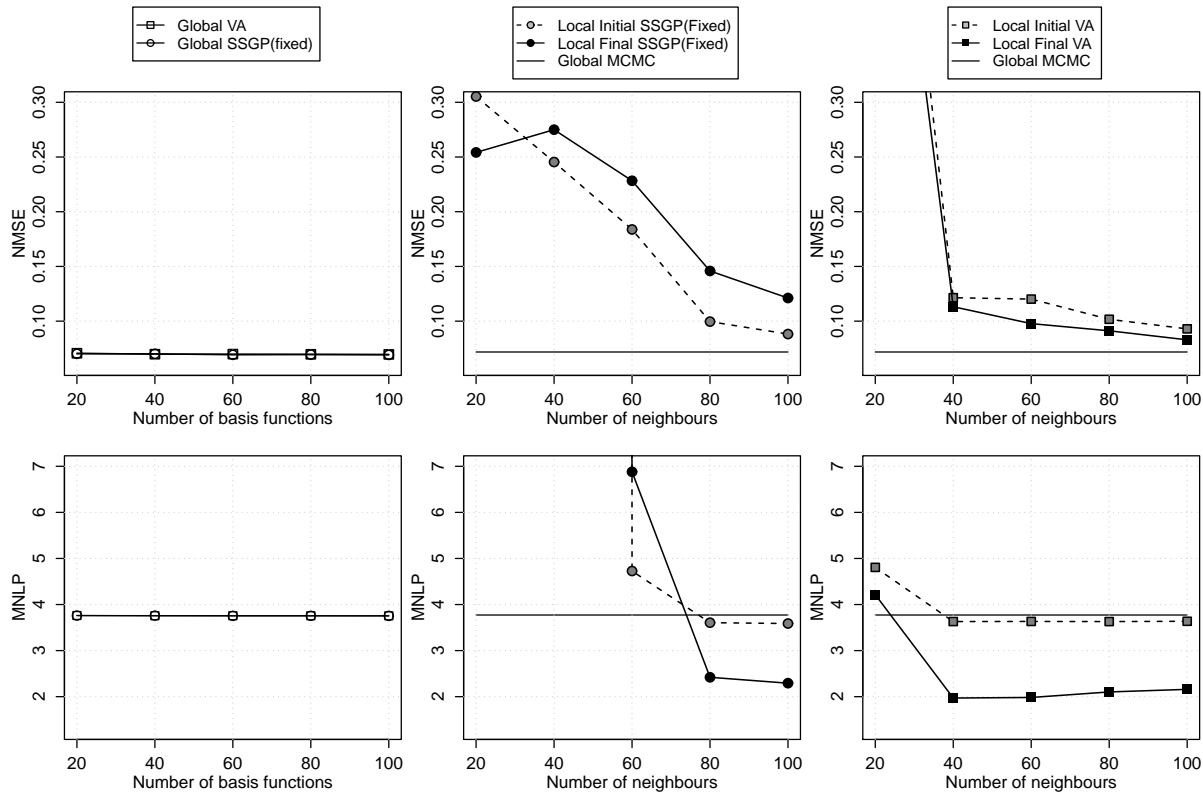} 
\caption{Rainfall-runoff simulated data. NMSE and MNLP values averaged over ten repetitions plotted against the number of basis functions (first column) and against the number of neighbours (second, third and fourth columns). Number of basis functions used in the local methods was 20.}
\label{rainfall12covMNLPNMSE}  
\end{figure*}

\subsection{Auto-MPG data} \label{Auto_section}

In this example, we consider the Automobile city-cycle fuel consumption in miles per gallon (Auto-MPG) data taken from the CMU Statistics library. This dataset was used in the 1983 American Statistical Association Exposition and is available at \url{http://archive.ics.uci.edu/ml/datasets.html}. The dataset contains 398 instances and nine attributes. \citet{Quinlan1993} used this data to predict the attribute ``MPG", which is the city-cycle fuel consumption in miles per gallon. The other eight attributes include two multi-valued discrete, four continuous attributes and two categorical variables. We drop the two categorical variables, car name and origin, and keep the four continuous attributes and two multi-valued discrete variables. Six of the data points are removed as they have missing entries in some of the input variables. We randomly select 80 data points as the test set, and use the remaining 312 data points as the training set.

\begin{figure*}[htb!]
\centering
\includegraphics[width=0.66\textwidth, angle=270]{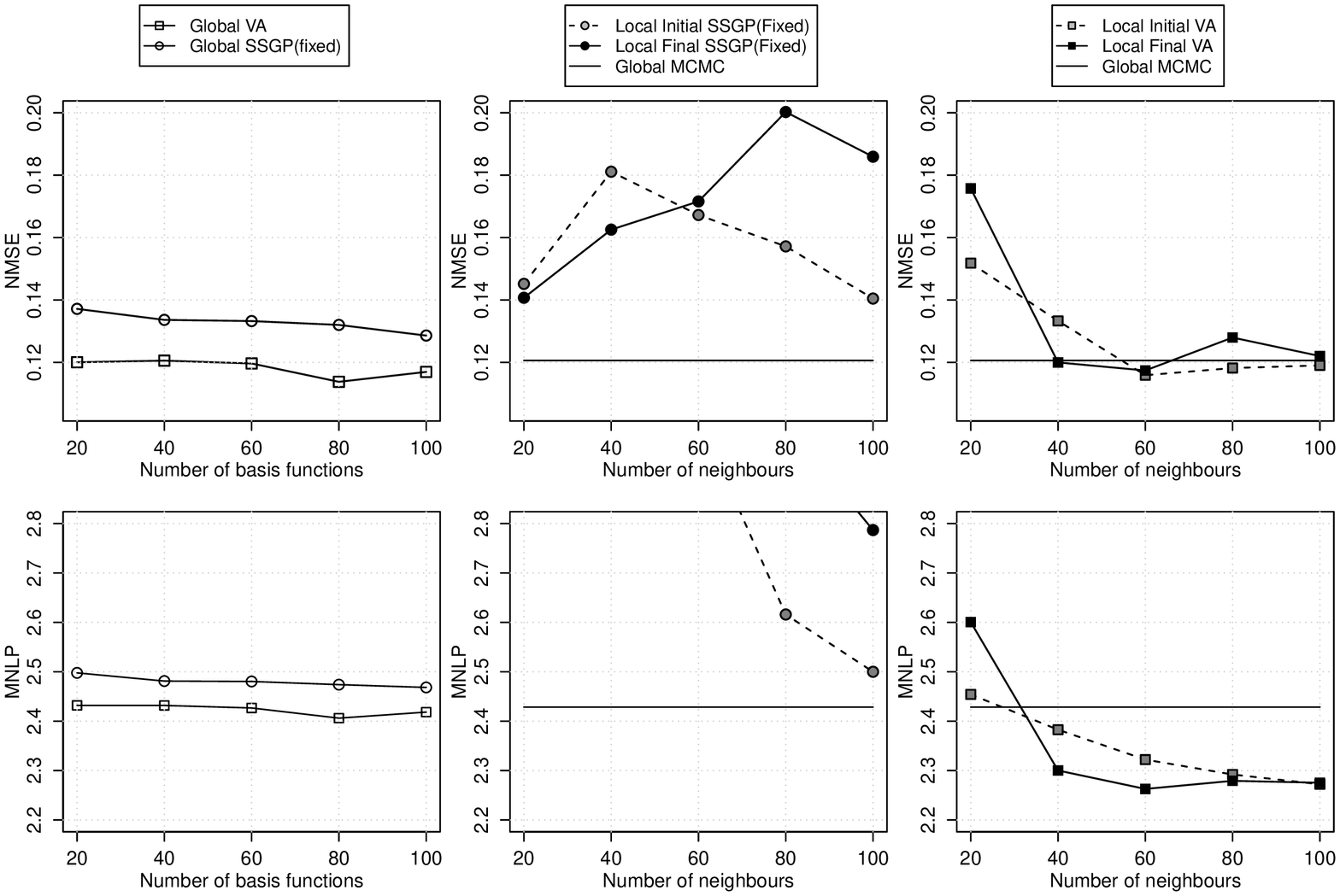} 
\caption{Auto-MPG data. NMSE and MNLP values averaged over ten repetitions plotted against the number of basis functions (first column) and against the number of neighbours (second and third  columns). Number of basis functions used in the local methods was 20.}
\label{auto}  
\end{figure*}

Figure \ref{auto} shows the NMSE and MNLP values averaged over ten repetitions. For the global SSGP (fixed) and VA methods, we observe slight improvements in both the NMSE and MNLP values as $m$ increases. The MNLP and NMSE values for MCMC and global VA are also better than for global SSGP (fixed). For the adaptive neighbourhood approach, we consider neighbourhoods of size $k = 20, 40, 60, 80, 100$, while fixing the number of basis function $m=20$. For local VA, the final fit is slightly better than the initial fit. There is an improvement brought about by adapting the neighbourhood as Figure \ref{auto} shows that the MNLP values of both the initial and final fits are lower than the MCMC method for neighbourhood size of 40, 60, 80 and 100. 

For local SSGP (fixed), we observe that their performance is worse than MCMC. Moreover, it seems that the initial fit is better than the final fit, for neighbourhood size of 80 and 100. This may be because the lengthscales are not accurate enough to be used for the final fit. We also examined the local SSGP approach with larger neighbourhood sizes of 150, 200 and 250. We found that, at a neighbourhood size of 150, the performance of the final fit of local SSGP (fixed) (MNLP and NMSE of 2.38 and 0.129 respectively) is slightly better than MCMC. Adapting the neighbourhood approach is still more apparent in the variational approach as it is able to achieve MNLP and NMSE value of 2.26 and 0.117 respectively at neighbourhood size of  60.

\subsection{Auto-MPG with simulated data}

We now consider the Auto-MPG data and look at the influence of irrelevant covariates on the model. This is again done by generating ten additional covariates artificially and randomly from the uniform distribution on the interval [0,1]. Once again, like the rainfall-runoff data, it seems that a larger neighbourhood is required to attain the best performance for the variational approach when irrelevant covariates are added. In this example, for the variational approach, we found that neighbourhood size of 100 produces the best peformance with MNLP and NMSE values of 2.30 and 0.127 respectively. Again, after examining the local SSGP (fixed) approach with larger neighbourhood sizes, we found that it attains the best performance (MNLP and NMSE of 2.43 and 0.141 respectively) at a neighbourhood size of 150.

\begin{figure*}[htb!]
\centering
\includegraphics[width=0.6\textwidth, angle=270]{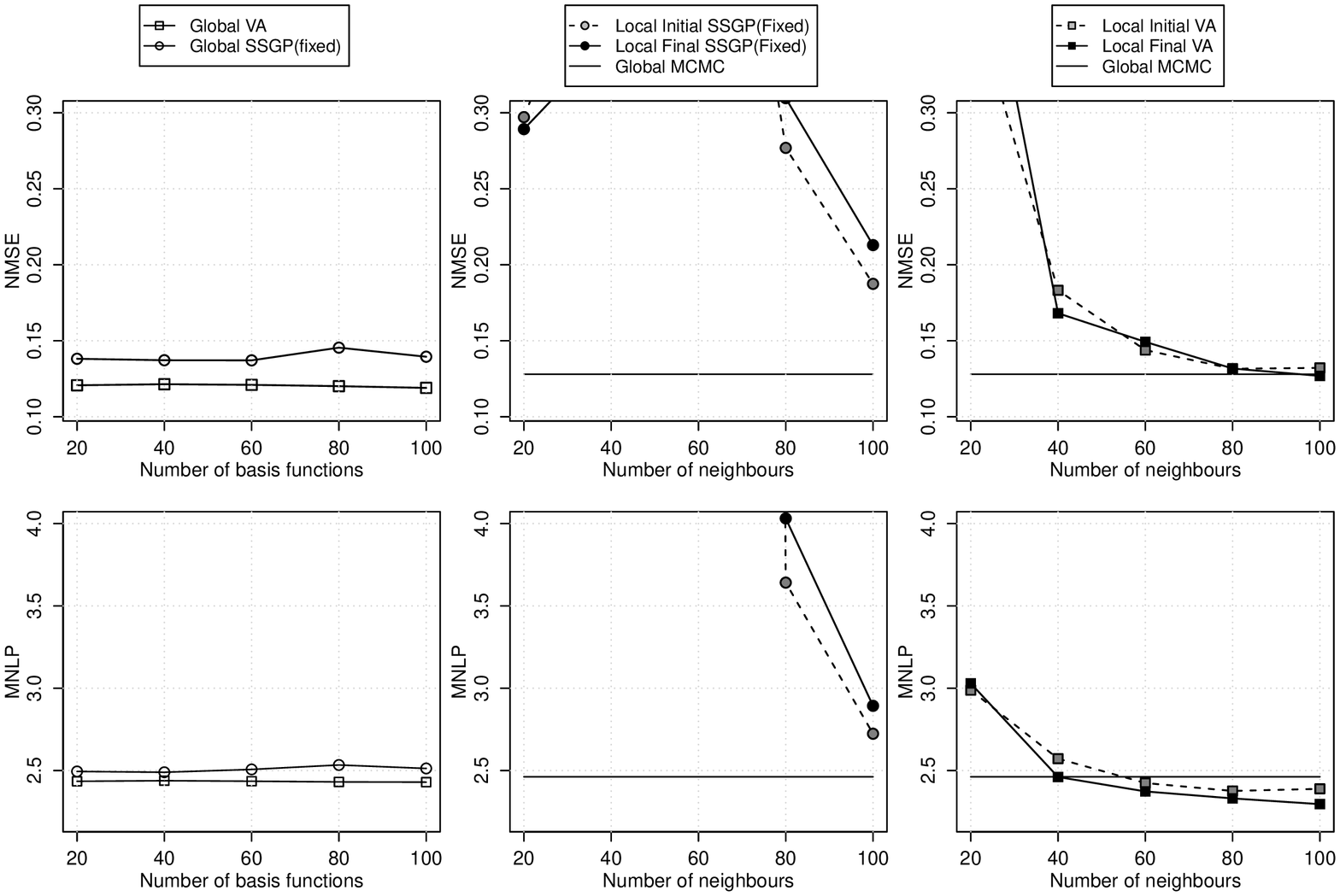} 
\caption{Auto-MPG data with 10 simulated covariates. NMSE and MNLP values averaged over ten repetitions plotted against the number of basis functions (first column) and against the number of neighbours (second and third columns). Number of basis functions used in the local methods was 20.}
\label{autoplus}  
\end{figure*}

\begin{figure*}[htb!]
\centering
\includegraphics[width=0.54\textwidth, height=0.9\textwidth, angle=270]{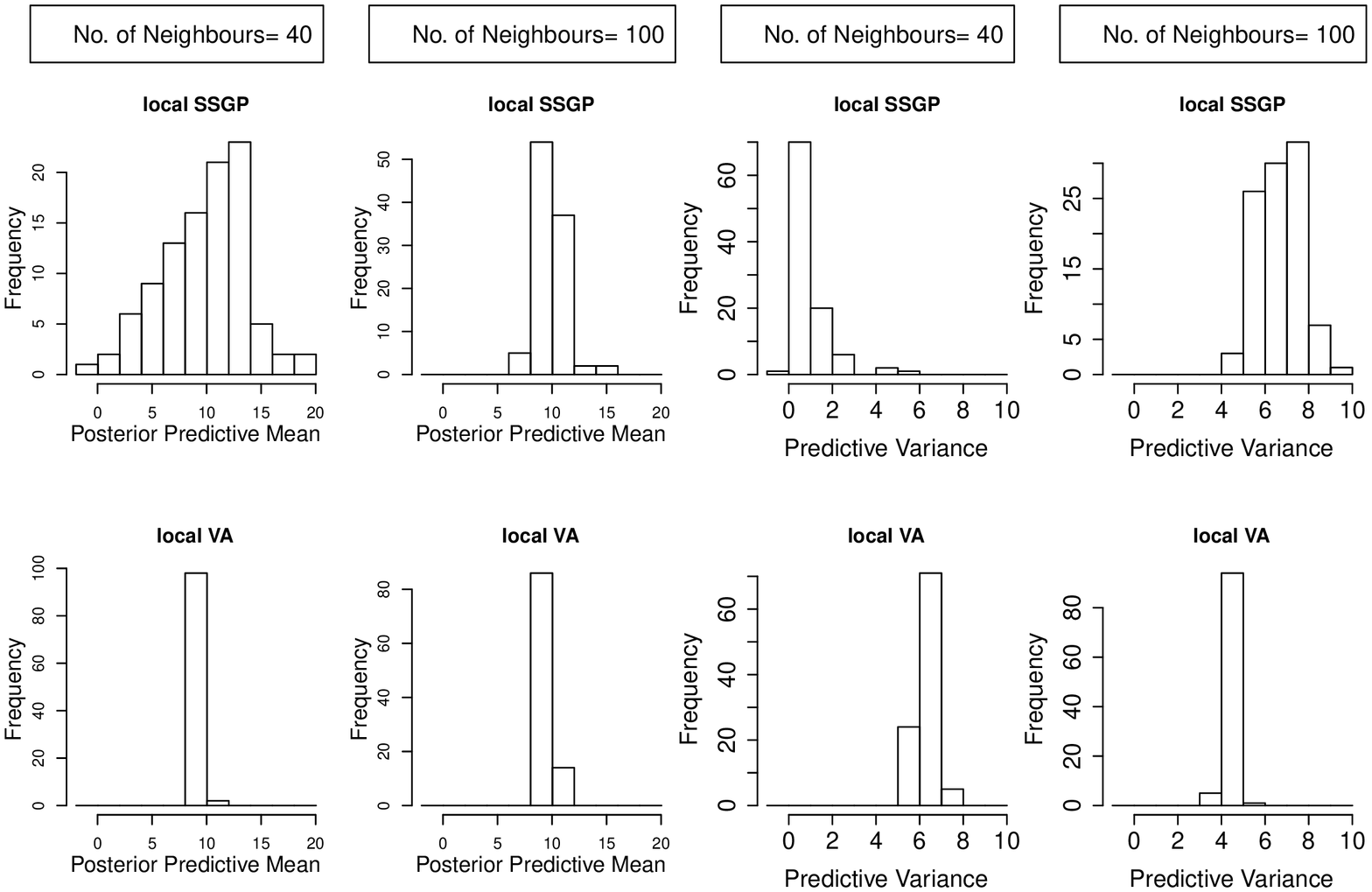} 
\caption{Auto-MPG data with 10 simulated covariates. Histogram of estimated posterior predictive means and variances over 100 repetitions. Number of basis functions used in the local methods was 20.}
\label{Posterior_mean} 
\end{figure*}

In order to explain why there is a difference in the stability of the adaptive neighbourhood approach between VA and SSGP (fixed), we examined the estimated predictive mean and variance for one test point from the Auto-MPG test set with 10 simulated irrelevant covariates. We implement the adaptive neighbourhoods approach based on just the initial fitting, which uses the shortest euclidean distance. Figure \ref{Posterior_mean} shows 100 posterior predictive means and variances from SSGP (fixed) and VA with the adaptive neighbourhood approach.  In the 100 replications, only the spectral points change. Since VA accounts for hyperparameter uncertainty, it is more robust towards the choice of spectral points. We observe that the posterior predictive means and variances are concentrated around a smaller range of values even when the size of neighbourhood is small. On the other hand, for local SSGP (fixed), the posterior predictive means vary more for different choices of the spectral points with the values ranging from 0 to 20 and with many of the posterior predictive variances small when the size of the neighbourhood is small. 

\section{Conclusion}\label{conclusion}
In this paper, we have presented a nonconjugate variational message passing algorithm for fitting sparse spectrum GP regression models where closed form updates are possible for all variational parameters, except for the evaluation of $\mathcal{H}(p,q,r)$. We note that $\mathcal{H}(p,q,r)$ can be evaluated very efficiently using quadrature and there is almost no computational overhead when compared to updates based on conditionally conjugate Inverse-Gamma priors for the variance parameters. However, half-Cauchy priors lead to much better predictive inference especially in the adaptive neighbourhood approach where the amount of training data is small. A Bayesian approach has been adopted for parameter estimation which allows covariance function hyperparameter uncertainty to be treated and empirical results suggest that this improves prediction (especially in the MNLP values) and prevents overfitting. We also propose a novel adaptive neighbourhood technique for obtaining predictive inference which is adept at handling data with nonstationarities and this approach can be extended to large data sets as well. The simulated data sets showed that weighting the dimensions according to the lengthscales estimated from an initial fit is very effective at downweighting variables of little relevance, leading to automatic variable selection and improved prediction. In addition, we introduce a technique for accelerating convergence in nonconjugate variational message passing by taking step sizes larger than one in the direction of the natural gradient of the lower bound. We do not attempt to search for the optimal step size but adopt an adaptive strategy that can be easily implemented, and empirical results indicate significant speed-ups. Algorithm \ref{Alg2} is thus an attractive alternative for fitting sparse spectrum GP regression models, which is stable, robust to overfitting for small data sets and capable of dealing with highly nonstationary data as well when used in combination with the adaptive neighbourhood approach.

\begin{acknowledgements}
We thank Lucy Marshall for supplying the rainfall-runoff data set. Linda Tan was partially supported as part of the Singapore Delft Water Alliance's tropical reservoir research programme. David Nott, Ajay Jasra and Victor Ong's research was supported by a Singapore Ministry of Education Academic Research Fund Tier 2 grant (R-155-000-143-112). We also thank the referees and associate editor for their comments which have helped improved the manuscript.
\end{acknowledgements}

\begin{appendices}
\renewcommand\thesection{Appendix \Alph{section}:}

\section{Derivation of $E_q(Z)$ and $E_q(Z^TZ)$}

\begin{lemma}
Suppose $\lambda \sim N(\mu,\Sigma)$ and $t_1$, $t_2$ are fixed vectors the same length as $\lambda$. Let $t_{12}^-=t_1-t_2$ and $t_{12}^+=t_1+t_2$, then
\begin{multline*}
E\{\cos(t_1^T\lambda)\cos(t_2^T\lambda)\} =\tfrac{1}{2}\left[\exp(-\tfrac{1}{2}{t_{12}^-}^T\Sigma t_{12}^-) \right.\\
 \left. \cdot \cos({t_{12}^-}^T\mu)  +\exp(-\tfrac{1}{2}{t_{12}^+}^T\Sigma t_{12}^+)\cos({t_{12}^+}^T\mu)\right]
 \end{multline*}
 \begin{multline*}
 E\{\sin(t_1^T\lambda)\sin(t_2^T \lambda)\} =\tfrac{1}{2}\left[\exp(-\tfrac{1}{2}{t_{12}^-}^T\Sigma t_{12}^-)\right. \\
  \left. \cdot \cos({t_{12}^-}^T\mu) -\exp(-\tfrac{1}{2}{t_{12}^+}^T\Sigma t_{12}^+)\cos({t_{12}^+}^T\mu)\right] 
 \end{multline*}
\begin{multline*}
E\{\sin(t_1^T\lambda)\cos(t_2^T \lambda)\} =\tfrac{1}{2}\left[\exp(-\tfrac{1}{2}{t_{12}^-}^T\Sigma t_{12}^-)  \right. \\
 \left. \cdot \sin({t_{12}^-}^T\mu)+\exp(-\tfrac{1}{2}{t_{12}^+}^T\Sigma t_{12}^+)\sin({t_{12}^+}^T\mu)\right]
\end{multline*}
By setting $t_2=0$ in the first and third expressions, we get 
\begin{equation*}
\begin{aligned}
E\{\cos(t_1^T\lambda)\} &=\exp(-\tfrac{1}{2}t_1^T\Sigma t_1)\cos(t_1^T\mu)\;\; \text{and} \\
E\{\sin(t_1^T\lambda)\} &=\exp(-\tfrac{1}{2}t_1^T\Sigma t_1)\sin(t_1^T\mu).
\end{aligned}
\end{equation*}
\end{lemma}

\begin{proof}
$E[\exp\{i\lambda^T(t_1-t_2)\}]=\exp\{i \mu^T (t_1-t_2)-
\tfrac{1}{2}(t_1-t_2)^T \Sigma (t_1-t_2)\}$
implies
\begin{equation}\label{e1} 
\begin{aligned}
E[\cos\{\lambda^T(t_1-t_2)\}] & = E\{\cos(t_1^T\lambda)\cos(t_2^T\lambda) \\
& \quad +\sin(t_1^T\lambda)\sin(t_2^T\lambda)\}  \\
& = \exp\{-\tfrac{1}{2}(t_1-t_2)^T \Sigma (t_1-t_2)\}  \\
& \quad \cdot \cos\{\mu^T (t_1-t_2)\}
\end{aligned}
\end{equation}
and
\begin{equation}\label{e2}
\begin{aligned}
E[\sin\{\lambda^T(t_1-t_2)\}] & = E\{\sin(t_1^T\lambda)\cos(t_2^T\lambda) \\
& \quad -\cos(t_1^T\lambda)\sin(t_2^T\lambda)\}  \\
& = \exp\{-\tfrac{1}{2}(t_1-t_2)^T \Sigma (t_1-t_2)\}  \\
& \quad \cdot \sin\{\mu^T (t_1-t_2)\}.
\end{aligned}
\end{equation}
Replacing $t_2$ by $-t_2$, we get
\begin{equation} \label{e3} 
\begin{aligned}
E[\cos\{\lambda^T(t_1+t_2)\}]& = E\{\cos(t_1^T\lambda)\cos(t_2^T\lambda) \\
& \quad -\sin(t_1^T\lambda)\sin(t_2^T\lambda)\}  \\
& = \exp\{-\tfrac{1}{2}(t_1+t_2)^T \Sigma (t_1+t_2)\} \\
& \quad  \cdot \cos\{\mu^T (t_1+t_2)\}
\end{aligned}
\end{equation}
and
\begin{equation}\label{e4}
\begin{aligned}
E[\sin\{\lambda^T(t_1+t_2)\}]& = E\{\sin(t_1^T\lambda)\cos(t_2^T\lambda) \\
& \quad +\cos(t_1^T\lambda)\sin(t_2^T\lambda)\} \\
& = \exp\{-\tfrac{1}{2}(t_1+t_2)^T \Sigma (t_1+t_2)\}  \\
& \quad \cdot  \sin\{\mu^T (t_1+t_2)\}.
\end{aligned}
\end{equation}
(\ref{e1})+(\ref{e3}) gives the first equation of the lemma, (\ref{e1})-(\ref{e3}) gives the second and (\ref{e2})+(\ref{e4}) gives the third. \QEDB
\end{proof}

Using Lemma 1, we have 
\begin{equation*}
E_q(Z)=[E_q(Z_1), \dots, E_q(Z_n)]^T,
\end{equation*}
where
\begin{equation*}
\begin{aligned}
E_q(Z_i^T) & = \big[ \exp(-\tfrac{1}{2}t_{i1}^T\Sigma_\lambda^q t_{i1})\cos(t_{i1}^T\mu_\lambda^q) ,\dots, \\
& \qquad \exp(-\tfrac{1}{2}t_{im}^T\Sigma_\lambda^q t_{im})\cos(t_{im}^T\mu_\lambda^q), \\
& \qquad \exp(-\tfrac{1}{2}t_{i1}^T\Sigma_\lambda^q t_{i1})\sin(t_{i1}^T\mu_\lambda^q),\dots, \\
& \qquad \exp(-\tfrac{1}{2}t_{im}^T\Sigma_\lambda^q t_{im})\sin(t_{im}^T\mu_\lambda^q) \big]
\end{aligned}
\end{equation*}
and $t_{ir}=s_r \odot x_i$ for $i=1,\dots,n$, $r=1,\dots,m$. We also have $E_q(Z^TZ)=\sum_{i=1}^n E_q(Z_iZ_i^T)$ where $E_q(Z_iZ_i^T)=\left[\begin{smallmatrix}
P_i & Q_i^T \\  Q_i & R_i  \end{smallmatrix}\right]$, where $P_i$, $Q_i$, $R_i$ are all $m\times m$ matrices and 
\begin{equation*}
\begin{aligned}
{P_i}_{rl} &=\tfrac{1}{2}\big\{\exp(-\tfrac{1}{2}{t_{irl}^-}^T\Sigma_\lambda^qt_{irl}^-) \cos({t_{irl}^-}^T) \mu_\lambda^q \\
& \qquad\qquad + \exp(-\tfrac{1}{2}{t_{irl}^+}^T\Sigma_\lambda^qt_{irl}^+) \cos({t_{irl}^+}^T) \mu_\lambda^q\big\},\\
{Q_i}_{rl} &=\tfrac{1}{2}\big\{-\exp(-\tfrac{1}{2}{t_{irl}^-}^T\Sigma_\lambda^qt_{irl}^-) \sin({t_{irl}^-}^T) \mu_\lambda^q \\
& \qquad\qquad + \exp(-\tfrac{1}{2}{t_{irl}^+}^T\Sigma_\lambda^qt_{irl}^+) \sin({t_{irl}^+}^T) \mu_\lambda^q\big\},\\
{R_i}_{rl} &=\tfrac{1}{2}\big\{\exp(-\tfrac{1}{2}{t_{irl}^-}^T\Sigma_\lambda^qt_{irl}^-) \cos({t_{irl}^-}^T) \mu_\lambda^q \\
&\qquad \qquad - \exp(-\tfrac{1}{2}{t_{irl}^+}^T\Sigma_\lambda^qt_{irl}^+) \cos({t_{irl}^+}^T) \mu_\lambda^q\big\},
\end{aligned}
\end{equation*}
$t_{irl}^-=t_{ir}-t_{il}$, $t_{irl}^+=t_{ir}+t_{il}$ for $r=1,\dots,m$, $l=1,\dots,m$.

\section{Derivation of lower bound}
From (\ref{LB}), the lower bound is given by 
\begin{equation*}
\mathcal{L}=E_q\{\log p(y,\theta)\}-E_q\{\log q(\theta)\}
\end{equation*}
where
\begin{align*}
E_q\{\log p(y,\theta)\}&= E_q\{\log p(y|\alpha,\lambda,\gamma)\} + E_q\{\log p(\alpha|\sigma)\} \\
&\quad +E_q\{\log p(\lambda)\}+E_q\{\log p(\sigma)\} \\
& \quad +E_q\{\log p(\gamma)\}, \\
E_q\{\log q(\theta)\} &= E_q\{\log q(\alpha)\} +E_q\{\log q(\lambda)\} \\
& \quad +E_q\{\log q(\sigma)\}+E_q\{\log q(\gamma)\}.
\end{align*}
The terms in the lower bound can be evaluated as follows:
\begin{multline*}
E_q\{\log p(y|\alpha,\beta,\lambda,\gamma)\}=-\tfrac{n}{2}\log (2\pi)-\tfrac{n}{2}E_q(\log\gamma^2) \\
-\tfrac{1}{2}\big[y^Ty-2y^TE_q(Z)\mu_\alpha^q+\text{tr}\{(\mu_\alpha^q{\mu_\alpha^q}^T+\Sigma_\alpha^q)E_q(Z^TZ)\}\big] \\
\cdot {\mathcal{H}(n,C_\gamma^q,A_\gamma^2)}/{\mathcal{H}(n-2,C_\gamma^q,A_\gamma^2)}
\end{multline*}
\begin{multline*}
E_q\{\log p(\alpha|\sigma)\}=-m\log (2\pi) -mE_q\{\log\sigma^2\} \\
+m\log m -\tfrac{m}{2}\tfrac{\mathcal{H}(2m,C_\sigma^q,A_\sigma^2)} {\mathcal{H}(2m-2,C_\sigma^q,A_\sigma^2)}\{{\mu_\alpha^q}^T\mu_\alpha^q+\text{tr}(\Sigma_\alpha^q)\}
\end{multline*}
\begin{multline*}
E_q\{\log p(\lambda)\}=-\tfrac{d}{2}\log (2\pi) -\tfrac{1}{2}\log|\Sigma_\lambda^0| \\ -\tfrac{1}{2}(\mu_\lambda^q-\mu_\lambda^0)^T{\Sigma_\lambda^0}^{-1}(\mu_\lambda^q-\mu_\lambda^0)-\tfrac{1}{2}\text{tr}({\Sigma_\lambda^0}^{-1}\Sigma_\lambda^q)
\end{multline*}
\begin{equation*}
E_q\{\log p(\sigma)\}=\log(2A_\sigma)-\log\pi-E_q\{\log(A_\sigma^2+\sigma^2)\}
\end{equation*}
\begin{equation*}
E_q\{\log p(\gamma)\}=\log(2A_\gamma)-\log\pi-E_q\{\log(A_\gamma^2+\gamma^2)\}
\end{equation*}
\begin{equation*}
E_q\{\log q(\alpha)\}=-m\log(2\pi)-\frac{1}{2}\log|\Sigma_\alpha^q|-m
\end{equation*}
\begin{equation*}
E_q\{\log q(\lambda)\}=-\tfrac{d}{2}\log (2\pi) -\tfrac{1}{2}\log|\Sigma_\lambda^q|-\tfrac{d}{2}
\end{equation*}
\begin{multline*}
E_q\{\log p(\sigma)\}=-C_\sigma^q \tfrac{\mathcal{H}(2m,C_\sigma^q,A_\sigma^2)}{\mathcal{H}(2m-2,C_\sigma^q,A_\sigma^2)} - 2mE_q\{\log\sigma\}\\
-\log \mathcal{H}(2m-2,C_\sigma^q,A_\sigma^2) -E_q\{\log(A_\sigma^2+\sigma^2)\}
\end{multline*}
\begin{multline*}
E_q\{\log p(\gamma)\} =-C_\gamma^q  {\mathcal{H}(n,C_\gamma^q,A_\gamma^2)} /{\mathcal{H}(n-2,C_\gamma^q,A_\gamma^2)}  \\
-\log \mathcal{H}(n-2,C_\gamma^q,A_\gamma^2)-nE_q\{\log\gamma\}-E_q\{\log(A_\gamma^2+\gamma^2)\}
\end{multline*}
Putting these terms together and making use of the updates in steps 5 and 6 of Algorithm \ref{Alg1} gives the lower bound in \eqref{LBexp}.

\section{Derivation of simplified updates in Algorithm 2}

It can be shown \citep[see][]{Wand2013,Tan2013} that the natural parameter of $q(\lambda)=N(\mu_\lambda^q,\Sigma_\lambda^q)$ is 
\begin{equation*}
\eta_\lambda=\begin{bmatrix}  -\frac{1}{2}D_d^T \text{vec}({\Sigma_\lambda^q}^{-1}) \\ {\Sigma_\lambda^q}^{-1}\mu_\lambda^q \end{bmatrix},
\end{equation*}
where $D_d$ is a unique $d^2 \times \tfrac{d}{2}(d+1)$ matrix that transforms $\text{vech}(A)$ into $\text{vec}(A)$ for any $d \times d$ symmetric square matrix $A$, that is, $D_d\text{vech}(A)=\text{vec}(A)$. We use $\text{vech}(A)$ to denote the $\tfrac{1}{2}d(d+1) \times 1$ vector obtained from $\text{vec}(A)$ by eliminating all supradiagonal elements of $A$. \cite{Magnus1988} is a good reference for the matrix differential calculus involved in the derivation below. From \eqref{step size update} and \citep[][pg. 7]{Tan2013}, we have 
\begin{multline} \label{f1}
\begin{bmatrix}  -\frac{1}{2}D_d^T \text{vec} \Big({{\Sigma_\lambda^q}^{(t)}}^{-1}\Big) \\ {{\Sigma_\lambda^q}^{(t)}}^{-1} {\mu_\lambda^q}^{(t)} \end{bmatrix}
 =(1-a_t) \\
 \cdot \begin{bmatrix}  -\frac{1}{2}D_d^T \text{vec}\Big({{\Sigma_\lambda^q}^{(t-1)}}^{-1}\Big) \\ {{\Sigma_\lambda^q}^{(t-1)}}^{-1} {\mu_\lambda^q}^{(t-1)} \end{bmatrix} \\
+ a_t \begin{bmatrix} D_d^T  & 0 \\
-2({{\mu_\lambda^q}^{(t-1)}}^T \otimes I){D_d^+}^T \negthinspace D_d^T & I 
\end{bmatrix}   \sum_{a \in N(\lambda)}
\begin{bmatrix}
\frac{\partial S_a}{\partial  \text{vec}(\Sigma_\lambda^q)} \\ 
\frac{\partial S_a}{\partial  \mu_\lambda^q} 
\end{bmatrix},
\end{multline}
where $\dfrac{\partial S_a}{\partial  \text{vec}(\Sigma_\lambda^q)}$ and $\dfrac{\partial S_a}{\partial  \mu_\lambda^q}$ are evaluated at
\begin{equation*}
\Sigma_\lambda^q={{\Sigma_\lambda^q}^{(t)}}^{-1} \;\; \text{and} \;\;  \mu_\lambda^q={\mu_\lambda^q}^{(t-1)}.
\end{equation*}  
Let 
\begin{equation*}
\sum_{a \in N(\lambda)} \frac{\partial S_a}{\partial  \text{vec}(\Sigma_\lambda^q)} = -\frac{1}{2}\text{vec}(G).
\end{equation*}
The first line of \eqref{f1}  simplifies to 
\begin{equation*}
\begin{aligned}
{{\Sigma_\lambda^q}^{(t)}}^{-1} & = (1-a_t) {{\Sigma_\lambda^q}^{(t)}}^{-1} + a_t G  \\
\Rightarrow {\Sigma_\lambda^q}^{(t)} & =\{(1-a_t) {{\Sigma_\lambda^q}^{(t)}}^{-1} + a_t G\}^{-1}.
\end{aligned}
\end{equation*}
The second line of \eqref{f1} gives 
\begin{equation*}
\begin{aligned}
{{\Sigma_\lambda^q}^{(t)}}^{-1} {\mu_\lambda^q}^{(t)} &=(1-a_t) {{\Sigma_\lambda^q}^{(t-1)}}^{-1} {\mu_\lambda^q}^{(t-1)} \\
& \quad + a_t G {\mu_\lambda^q}^{(t-1)} + a_t \sum_{a \in N(\lambda)}\frac{\partial S_a}{\partial  \mu_\lambda^q} \\
& =  {{\Sigma_\lambda^q}^{(t)}}^{-1}   {\mu_\lambda^q}^{(t-1)} + a_t \sum_{a \in N(\lambda)}\frac{\partial S_a}{\partial  \mu_\lambda^q}  \\
\Rightarrow   {\mu_\lambda^q}^{(t)} &=  {\mu_\lambda^q}^{(t-1)}   + a_t  {\Sigma_\lambda^q}^{(t)}  \sum_{a \in N(\lambda)}\frac{\partial S_a}{\partial  \mu_\lambda^q}.
\end{aligned}
\end{equation*}

\end{appendices}

\end{document}